\newif\ifshort
\newif\ifappendix

\documentclass[11pt,a4paper,pagebackref]{article}
\usepackage[utf8]{inputenc}
\usepackage{microtype}
\usepackage{amsmath}
\usepackage{amsfonts}
\usepackage{nicefrac}
\usepackage{amssymb}
\usepackage{amsthm}
\usepackage{mathtools}
\usepackage{tikz}
\usetikzlibrary{calc}
\usepackage{pgfplots}
\pgfplotsset{compat=newest}
\usepackage{pgfplotstable}
\usepackage{xcolor}
\usepackage{authblk}
\usepackage[margin=2.5cm,a4paper]{geometry}

\pgfplotsset{
    discard if not/.style 2 args={
        x filter/.code={
            \edef\tempa{\thisrow{#1}}
            \edef\tempb{#2}
            \ifx\tempa\tempb
            \else
                
            \fi
        }
    }
}
\pgfmathdeclarefunction{lg2}{1}{%
    \pgfmathparse{ln(#1)/ln(2)}%
}
\pgfmathdeclarefunction{lg10}{1}{%
    \pgfmathparse{ln(#1)/ln(10)}%
}

\definecolor{r}{rgb}{1.0, 0.4, 0.4}
\definecolor{r0}{rgb}{1.0, 0.7, 0.4}
\definecolor{r1}{rgb}{1.0, 0.4, 0.0}
\definecolor{r2}{rgb}{0.8, 0, 0.4}
\definecolor{b}{rgb}{0.4, 0.4, 1.0}
\definecolor{b0}{rgb}{0.4, 0.7, 1.0}
\definecolor{b1}{rgb}{0.0, 0.4, 1.0}
\definecolor{b2}{rgb}{0.4, 0.0, 0.8}
\definecolor{g}{rgb}{0.4, 1.0, 0.4}

\usepackage{tabularx}
\usepackage{booktabs}
\usepackage{multirow}
\usepackage{paralist}

\usepackage{algorithm2e}
\DontPrintSemicolon

\usepackage[pagebackref,pdfdisplaydoctitle,menucolor=orange!40!black,filecolor=magenta!40!black,urlcolor=blue!40!black,linkcolor=red!40!black,citecolor=green!40!black,colorlinks]{hyperref}
\usepackage[square,numbers]{natbib}
\usepackage[nameinlink,sort&compress,capitalize]{cleveref}

\usepackage[backgroundcolor=gray!10,textsize=footnotesize]{todonotes}

\ifshort{}

\fi{}
\newcommand{\appsymb}{$\bigstar$}
\newcommand{\appref}[1]{\ifshort{}{\hyperref[proof:#1]{\appsymb}}\fi{}}

\usepackage{etoolbox}

\newcommand{\appendixsection}[1]{%
	\ifshort{}\gappto{\appendixProofText}{\section{Additional Material for Section~\ref{#1}}\label{app:#1}}\fi{}
}

\usepackage{thmtools}
\newtheorem{theorem}{Theorem}
\newtheorem{lemma}[theorem]{Lemma}
\newtheorem{observation}[theorem]{Observation}

\newtheorem{proposition}[theorem]{Proposition}

\theoremstyle{definition}

\declaretheorem[style=definition,name=Construction,qed=$\diamond$]{construction}

\crefname{rrule}{Rule}{Rules}
\crefname{construction}{Construction}{Constructions}
\crefname{figure}{Figure}{Figures}

\newcommand{\prob}[1]{\textnormal{\textsc{#1}}}
\newcommand{\probDef}[3]{
	\begin{center}
	\begin{minipage}{0.95\columnwidth}
		\noindent
		\textsc{#1}
		\vspace{5pt}\\
		\setlength{\tabcolsep}{3pt}
		\begin{tabularx}{\textwidth}{@{}lX@{}}
			\textbf{Input:}     & #2 \\
			\textbf{Question:}  & #3
		\end{tabularx}
	\end{minipage}
	\end{center}
}

\DeclarePairedDelimiterX{\abs}[1]{\lvert}{\rvert}{#1}
\DeclarePairedDelimiterX{\norm}[1]{\lVert}{\rVert}{#1}
\DeclarePairedDelimiterX{\ceil}[1]{\lceil}{\rceil}{#1}

\newcommand{\NN}{\mathbb{N}}
\newcommand{\Nzero}{\mathbb{N}_0}

\newcommand{\QQ}{\mathbb{Q}}

\newcommand{\Wone}{\ensuremath{\mathrm{W[1]}}}

\newcommand{\NP}{\ensuremath{\mathrm{NP}}}

\newcommand{\bigO}{\mathcal{O}}
\newcommand{\yes}{\emph{yes}}
\newcommand{\no}{\emph{no}}

\newcommand{\oneto}[1]{[ #1 ]} %

\newcommand{\eps}{\varepsilon}

\DeclareMathOperator{\ed}{\#ed}
\newcommand{\diff}{\ensuremath{\Delta_{\mathrm{ed}}}}
\newcommand{\diffnorm}{\ensuremath{\delta_{\mathrm{norm}}}}
\newcommand{\Diffnorm}{\ensuremath{\Delta_{\mathrm{norm}}}}

\newcommand{\Diffavg}{\ensuremath{\Delta_{\mathrm{avg}}}}

\newcommand{\FCE}{\prob{Mod\-i\-fi\-ca\-tion-Fair Cluster Editing}}
\newcommand{\FCC}{\prob{Mod\-i\-fi\-ca\-tion-Fair Cluster Completion}}
\newcommand{\FCD}{\prob{Mod\-i\-fi\-ca\-ti\-on-Fair Cluster Deletion}}
\newcommand{\CE}{\prob{Cluster Editing}}
\newcommand{\CC}{\prob{Cluster Completion}}

\newcommand{\CTA}{\prob{Cluster Transformation by Edge Addition}}

\newcommand{\smallbinom}[2]{\Bigl(\begin{array}{@{}c@{}}#1\\#2\end{array}\Bigr)}
\newcommand{\ctwo}[1]{\ensuremath \smallbinom{#1}{2}}

\title{\Large\bf Modification-Fair Cluster Editing}

\author{Vincent Froese}
\author{Leon Kellerhals}
\author{Rolf Niedermeier}

\affil{\small
  Technische Universit\"at Berlin, Algorithmics and Computational Complexity, Berlin, Germany\protect\\
  \texttt{\{vincent.froese,\,leon.kellerhals\}@tu-berlin.de}}

\date{}

\begin{document}

\maketitle

\begin{abstract}
	The classic \CE{} problem (also known as \textsc{Correlation Clustering}) asks to transform a given graph into a disjoint union of cliques (clusters) by a small number of edge modifications.
	When applied to vertex-colored graphs (the colors representing subgroups), standard algorithms for the NP-hard \CE{} problem may yield  solutions that are biased towards subgroups of data (e.g., demographic groups), measured in the number of modifications incident to the members of the subgroups.
	We propose a modification fairness constraint which ensures that the number of edits incident to each subgroup is proportional to its size.
	To start with, we study %
	\textsc{Modification-Fair Cluster Editing} for graphs with two vertex colors.
	We show that the problem is NP-hard even if one may only \emph{insert} edges \emph{within} a subgroup; note that in the classic ``non-fair'' setting, this case is trivially polynomial-time solvable.
	However, in the more general \emph{editing} form, the modification-fair variant remains fixed-parameter tractable with respect to the number of edge edits.
	We complement these and further theoretical results with an empirical analysis of our model on real-world social networks where we find that the price of modification-fairness is surprisingly low, that is, the cost of optimal modification-fair solutions differs from the cost of optimal ``non-fair'' solutions only by a small percentage.
\end{abstract}

\section{Introduction}
In recent years, fairness in algorithmic problems has become 
a profoundly studied topic, particularly so in machine learning and related areas.
Clustering problems are fundamental in unsupervised learning 
and optimization in general. In this work, we focus on graph-based 
data clustering, and therein on one of the most basic and best studied 
problems, \CE{} (also known as \textsc{Correlation Clustering}).
The goal is to cluster the vertices into a set of disjoint cliques by (few) edge modifications, that is, edge deletions or insertions.
In the context of fairness, each vertex belongs to a certain subgroup within a social network (e.g.~gender or nationality)
and the goal is to find a solution that guarantees some ``fairness'' with respect to the considered subgroups.
Previous works~\cite{ahmadi2022fair,AEKM20,FM21,ahmadian2022improved} mainly focus on ``output-oriented'' fairness, that is,
the fairness is defined by looking at the resulting clusters,
enforcing that within each cluster, the number of vertices of each group is proportional to the overall number of vertices of the group.
This kind of fairness, while prudent in some scenarios, may be inapt in other contexts, e.g., political districting.

Our main conceptual contribution is to introduce a fairness concept that is
not modeling the fairness of the resulting clusters, but rather the fairness of the clustering process.
In our case, this means that each group should be affected by roughly the same (proportionally to its size) number of edge modifications.
This is motivated as follows: The edge modifications cause some \emph{distortion} of the true social network.
If the distortion for one group is significantly higher than for the others, then this can lead to a systematic bias in any further analysis of the cluster graph.
Hence, this distortion should be proportionally distributed among the groups in order not to yield wrong (biased) conclusions from the resulting clustering.
Imagine a collaboration graph where vertices are researchers from different countries (see \cref{fig:example} for an example).
The five modifications shown in \cref{fig:example}~(b) yield a solution where the number of blue and red vertices per cluster is well balanced; thus the transformation is fair in the ``output-oriented'' fairness setting.
However, most modifications are incident to blue vertices.
The resulting cluster graph might suggest that the researchers from the blue country are barely collaborating with each other but rather with researchers from the red country --- this does not really reflect the ground truth.
The modifications shown in \cref{fig:example}~(c) are more balanced between blue and red vertices.

To mitigate such possible bias as described above, we introduce a colored version of the well-studied NP-hard \CE{} problem, 
where now the criterion of having fair modification cost yields a \emph{process-oriented fairness} concept.
Our modification fairness for \CE{} aims at balanced average distortion among the groups and is similar in spirit to the socially fair variants of $k$-means and~$k$-median~\cite{ABV21,GSV21}, where the maximum average (representation) cost of any group is minimized.
Of course, fairness might come at a price, in that more edge modifications might be required to achieve modification-fair solutions and more computation time might be required to find these.
We perform both a theoretical (algorithms and complexity) and an empirical study.
In a nutshell, we show that our new problem \FCE{} seems computationally slightly harder than \CE{},
but our experimental studies also indicate that the ``price of fairness''  (that is, how many more edits are needed compared to the classic, ``colorblind'' case) is relatively low if one does not aim for perfect fairness.

\begin{figure}
	\centering
	\begin{tikzpicture}[
		colornode/.style = {
			circle,
			draw=#1!70!black,
			very thick,
			fill=#1,
			inner sep=3.0pt,
		},
		nonedge/.style = {thick},
		addedge/.style = {ultra thick,draw=g!80!black},
		deledge/.style = {ultra thick,loosely dotted,draw=g!80!black},
	]

	\def\xscopedist{3.5cm}

	\newcommand{\thenodes}[1]{
		\node at (0,5) {\small (#1)};
		
		\node[colornode=b] at (0,4) (b1) {};
		\node[colornode=b] at (0,3) (b2) {};
		\node[colornode=b] at (0,2) (b3) {};
		\node[colornode=b] at (0,1) (b4) {};

		\node[colornode=r] at (1.3,5) (r1) {};
		\node[colornode=r] at (1.3,4) (r2) {};
		\node[colornode=r] at (1.3,3) (r3) {};
		\node[colornode=r] at (1.3,2) (r4) {};
		\node[colornode=r] at (1.3,1) (r5) {};
	}

	\begin{scope}
		\thenodes{a}
		\draw
			(b1) edge[nonedge] (r1)
			(b1) edge[nonedge] (r2)
			(r1) edge[nonedge] (r2)
			(b1) edge[nonedge] (b2)
			(b1) edge[nonedge,bend right=30] (b3)
			(b2) edge[nonedge] (b3)
			(b3) edge[nonedge] (b4)
			(b2) edge[nonedge] (r3)
			(b3) edge[nonedge] (r4)
			(r3) edge[nonedge] (r4)
			(b4) edge[nonedge] (r5);

	\end{scope}

	\begin{scope}[xshift=\xscopedist]
		\thenodes{b}
		\draw
			(b1) edge[nonedge] (r1)
			(b1) edge[nonedge] (r2)
			(r1) edge[nonedge] (r2)
			(b1) edge[deledge] (b2)
			(b1) edge[deledge,bend right=30] (b3)
			(b2) edge[nonedge] (b3)
			(b3) edge[deledge] (b4)
			(b2) edge[nonedge] (r3)
			(b3) edge[nonedge] (r4)
			(r3) edge[nonedge] (r4)
			(b4) edge[nonedge] (r5)
			(b2) edge[addedge] (r4)
			(b3) edge[addedge] (r3);

	\end{scope}

	\begin{scope}[xshift=2*\xscopedist]
		\thenodes{c}
		\draw
			(b1) edge[deledge] (r1)
			(b1) edge[deledge] (r2)
			(r1) edge[nonedge] (r2)
			(b1) edge[nonedge] (b2)
			(b1) edge[nonedge,bend right=30] (b3)
			(b2) edge[nonedge] (b3)
			(b3) edge[deledge] (b4)
			(b2) edge[deledge] (r3)
			(b3) edge[deledge] (r4)
			(r3) edge[nonedge] (r4)
			(b4) edge[nonedge] (r5);

	\end{scope}

	\end{tikzpicture}
	\caption{
		An exemplary graph $G$ with blue (dark) and red (light) vertices (a) and two transformations of $G$ into a cluster graph (b), (c).
		Inserted edges are marked green (thick), deleted edges are green and dashed.
		(b)~A transformation of minimum size with five modifications.
		Eight modifications are incident to blue, two modifications are incident to red, so the average number of modifications to blue (red) vertices is $\nicefrac{8}{4}$ ($\nicefrac{2}{5}$), and the difference is~$\nicefrac{8}{5}$.
		(c)~Another minimum-size transformation in which the modifications are more balanced between blue and red (difference $\nicefrac{7}{10}$).
	}
	\label{fig:example}
\end{figure}

\paragraph{Related work.}
For a thorough review on fairness in the context of machine learning 
we refer to the survey by~\citet{mehrabi2022bias}.
Closest to our work in terms of the underlying clustering 
problem are studies on fair \textsc{Correlation Clustering}~\cite{ahmadi2022fair,AEKM20,FM21,ahmadian2022improved}. 
These works focus an output-oriented fairness, that is,
proportionality of the clusters.
Facing the \NP-hardness of the problem, these works mainly study polynomial-time approximation algorithms (while we focus on exact solvability).

\citet{chierichetti2017clustering} were the first to study fairness in the context of clustering, studying $k$-median and $k$-center problems.
The works by \citet{ABV21} and \citet{GSV21} for $k$-means and~$k$-median clustering are closest to our fairness concept.
There are numerous further recent works studying fairness for clustering problems~\cite{AEKMMPVW20,BFS21,BFGPS21,CN21,MV20,VY21}.
For a general account on classic \CE{}, we refer to the survey of \citet{BB13}.

\paragraph{Our contributions.}
We introduce \FCE, reflecting a process-oriented fairness criterion in graph-based data clustering:
instead of looking at the outcome, we consider the modification process that yields the clustering.
Here we demand that the average number of modifications at a vertex is balanced among the groups (we focus on two groups).
We parameterize our fairness constraint by the difference between these averages.
For a formal definition of \FCE{}, we refer to the next paragraph.
\cref{tab:results} gives an overview over our theoretical contributions to \FCE{}; the corresponding results are in \cref{sec:complexity}.
\begin{table}[t]
	\centering
	\caption{%
		Our theoretical results for \FCE{} and its restrictions which allow only insertions (\prob{Completion}) or deletions (\prob{Deletion}).
		Here, we denote
		by~$n$ the number of vertices,
		by~$m$ the number of edges,
		by~$\delta$ the fairness constraint,
		by~$k$ the number of modifications, and
		by~$\mu$ the number of mono-colored modifications, i.e., the number of modifications between same-colored endpoints.
		\newline
		$^\dagger$\,(even if only mono-colored modifications are allowed)
		$^\ddagger$\,(even if only one vertex is red)
	}
	\begin{tabular}{@{} r l l l@{}}\toprule
		\prob{Modification-Fair Cluster} $\ldots$	& {Complexity and running time} & Ref.\\
		\midrule
		\midrule
		$\ldots$ \prob{Completion}			& \NP-hard$^{\dagger}$ for any~$\delta \in \bigO(1)$ & \Cref{thm:FCC-NP-hard}\\
		\midrule
		$\ldots$ \prob{Deletion}			& \NP-hard$^{\dagger,\ddagger}$ for any~$\delta \ge 0$ & \Cref{thm:fce-hard}\\
		\midrule
		\multirow{3}{*}{$\ldots$ \prob{Editing}}	& \NP-hard$^{\dagger,\ddagger}$ for any~$\delta \ge 0$ & \Cref{thm:fce-hard}\\
		& $n^{\bigO(\mu)}$ (randomized)			& \Cref{thm:fpt}\\
		& $2^{\bigO(k \log k)}\cdot (n+m)$		& \Cref{thm:fpt}\\
		\bottomrule
	\end{tabular}
	\label{tab:results}
\end{table}
Among other results, we show that \FCE{} remains \NP-hard even if only edge insertions are allowed (in the classic, colorblind variant, this case is trivially polynomial-time).
This requires proving a related number problem to be \NP-hard, which is deferred to \cref{sec:transform}.
Moreover, we show the \NP-hardness of very restricted cases of the general editing version and provide conditional running time lower bounds.
On the positive side, we devise a randomized polynomial-time algorithm for the case that one modifies constantly many mono-colored edges (edges whose both endpoints have the same color).
Moreover, we show the problem to be fixed-parameter tractable with respect to the overall number of edge modifications.
On the empirical side (\cref{sec:experiments}), we demonstrate that while typically computationally hard(er) to find, ``fair solutions'' seem not much more expensive than conventional ones.

\paragraph{Problem definition and initial observations.}
Recall that a graph is a \emph{cluster graph} if and only if each of its connected components is a clique, that is, a completely connected graph.
The family of cluster graphs is also characterized as those graphs that do not contain a~$P_3$ (a path on three vertices) as an induced subgraph.
In \CE{}, we are given a graph~$G$ and an integer~$k \in \Nzero$,
and we are asked whether there is an \emph{edge modification set}~$S \subseteq \binom{V(G)}{2}$ of size at most~$k$
such that the graph~$G_S$ with vertex set~$V(G_S)\coloneqq V(G)$ and edge set~$E(G_S) \coloneqq (E(G) \setminus S) \cup (S \setminus E(G))$ is a cluster graph.
We say that $S$ \emph{transforms} $G$ into~$G_S$.

In our setting, the vertices in~$G$ are colored either \emph{red} or \emph{blue}, i.e., $V(G) = R \uplus B$.
For an edge modification set~$S \subseteq \binom{V(G)}{2}$,
we define $\ed_S(v)\coloneqq |\{e\in S\mid v\in e\}|$ to be the number of edge modifications incident to a vertex~$v$ (that is, the degree of~$v$ in the modification graph, whose edge set is~$S$).
Then
\begin{equation}
	\label{eq:diff}
	\diff(S)\coloneqq \left|\frac{\sum_{v\in R}\ed_S(v)}{|R|}-\frac{\sum_{v\in B}\ed_S(v)}{|B|}\right|
\end{equation}
is the difference of the average numbers of modifications at a red vertex and a blue vertex.

\probDef{\FCE}
{A graph~$G$ with~$V(G) = R \uplus B$, $k\in\NN$, and~$\delta\in\QQ^+$.}
{Is there an edge modification set~$S\subseteq \binom{V(G)}{2}$ with~$|S|\le k$ and $\diff(S) \le \delta$ that transforms~$G$ into a cluster graph?}

Analogously, we define the variants \FCC{} and \FCD{} in which~$S$ may only add edges (i.e., $S \subseteq \binom{V(G)}{2}\setminus E(G)$) and delete edges (i.e., $S \subseteq E(G)$), respectively.

We immediately observe some simple upper bounds on~$\diff$.

\begin{observation}
	\label{obs:bounds}
	For every edge modification set~$S$, the following upper bounds hold:
	\begin{inparaenum}[(i)]
		\item $\diff(S) \le |S|$;
		\item $\diff(S) \le |V(G)|-1$;
		\item $\diff(S) \le \nicefrac{2|S|}{\min\{|R|,|B|\}}$.
	\end{inparaenum}
\end{observation}

\begin{proof}
	The bounds $(i)$ and~$(ii)$ are trivial upper bounds on the maximum and thus also average number of edge modifications at any vertex.
	Bound~$(iii)$ holds as~$\diff(S)$ is at most
	\[
		\max \left\{ \frac{\sum_{v \in R} \ed_S(v)}{\abs{R}}, \frac{\sum_{v \in B} \ed_S(v)}{\abs{B}} \right\}
			\le \max \left\{ \frac{2\abs{S}}{\abs{R}}, \frac{2\abs{S}}{\abs{B}} \right\}
			= \frac{2\abs{S}}{\min\{\abs{R}, \abs{B}\}}.
	\]
	This bound is met when all endpoints of the modifications carry the less frequent color.
\end{proof}

By \cref{obs:bounds}, if~$\delta \ge \min(k,|V|-1,\allowbreak \nicefrac{2k}{\min\{|R|,|B|\}})$, then \FCE{}
is simply the standard, ``colorblind'', \CE{}.

We remark that that our problem definition allows the modification-fair edge modification set to be a non-minimal edge modification set.
If one seeks the most fair \emph{and} minimal edge modification set (of size at most $k$), then this can simply be computed with the standard $P_3$-branching algorithm~\cite{Cai96} which enumerates all solutions of size at most $k$.

\paragraph*{Parameterized Complexity.}
Finally, we recall some basic (parameterized) complexity concepts.
A parameterized problem is \emph{fixed-parameter tractable} if there exists an algorithm
solving any instance~$(x,p)$ ($x$ is in the input instance and $p$~is some parameter---in our case it will be the number~$k$ of edge modifications) in $f(p)\cdot|x|^{\bigO(1)}$ time, where~$f$ is a computable function solely depending on~$p$.
The class~XP contains all parameterized problems which can be solved in polynomial time if the parameter~$p$ is a constant, that is,
in $f(p)\cdot |x|^{g(p)}$~time.
\ifshort{}
The Exponential Time Hypothesis (ETH) claims that there exists a constant~$c > 0$ such that the \textsc{3-SAT} problem cannot be solved in~$\bigO(2^{cn})$ time.
\else{}
The Exponential Time Hypothesis (ETH) claims that the \textsc{3-SAT} problem cannot be solved in subexponential time in the number~$n$ of variables of the Boolean input formula.
That is, there exists a constant~$c > 0$ such that \textsc{3-SAT} cannot be solved in~$\bigO(2^{cn})$ time.
\fi{}
\ifshort{}It\else{}The ETH\fi{} is used to prove conditional running time lower bounds, for example, it is known that one cannot find a clique clique of size~$s$ in an~$n$-vertex graph in $\rho(s)\cdot n^{o(s)}$ time for any function~$\rho$, unless the ETH fails~\cite{CHKX06}.

\section{Modification Fairness: Complexity}
\label{sec:complexity}

We explore the algorithmic complexity of \FCE{} and compare it to its ``colorblind'' counterpart
\CE{} and its restrictions
which either only allow edge deletions (\textsc{Cluster Deletion}) or insertions (\textsc{Cluster Completion}).

First, we show that even restricted special cases of \FCE{} remain NP-hard.
Notably, the corresponding polynomial-time many-one reductions also lead to ETH-based running time lower bounds.

\begin{theorem}\label{thm:fce-hard}
	\FCE{} and \FCD{} are \NP-hard for arbitrary~$\delta\ge 0$ and solvable neither in $2^{o(k)}\cdot |V(G)|^{\bigO(1)}$ nor in $2^{o(\abs{V(G)} + \abs{E(G)})}$ time unless the ETH fails. This also holds 
	\begin{enumerate}[(i)]
		\item if only mono-colored edge modifications are allowed or
		\item if there is only one red vertex.
	\end{enumerate}
\end{theorem}

Both cases use similar reductions, based on the following \NP-hardness result by \citet{KU12} for standard \CE{}.

\begin{proposition}[\cite{KU12}]
	\label{prop:ce-hard}
	\CE{} is \NP-hard, and, assuming the ETH, is neither solvable in $2^{o(k)}\cdot |V(G)|^{\bigO(1)}$, nor in $2^{o(\abs{V(G)})}$, nor in~$2^{o(\abs{E(G)})}$ time, even if all of the following holds:
	\begin{enumerate}[(i)]
		\item all modifications are deletions;
		\item every solution has size at least~$k$;
		\item the graph has less than~$2k$ vertices;
		\item the graph has maximum degree six (and also contains vertices of degree exactly six);
		\item every solution deletes at most four edges incident to any vertex.
	\end{enumerate}
\end{proposition}

We now provide the construction for \cref{thm:fce-hard}(i).
While we construct an instance of \FCE{} in the following, we will later see that we can use the same construction for the \prob{Deletion} variant.

\begin{construction}[for \cref{thm:fce-hard}(i)]
	\label{constr:hard-1}
	Let~$I = (G, k)$ be an instance of \CE{}.
	We may assume that~$I$ has the properties listed in \cref{prop:ce-hard}.
	We construct an instance~$I' = (G', k', \delta)$ of \FCE{} as follows.
	The graph~$G'$ contains a copy of~$G$ with all vertices colored blue.
	Additionally, $G'$ contains~$3k$ red vertices which form~$k$ disjoint~$P_3$s, i.e., paths on three vertices.
	Moreover, we add~$\max\{\abs{V(G)}, 3k\} - 3k$ isolated red vertices and~$\max \{\abs{V(G)}, 3k\} - \abs{V(G)}$ isolated blue vertices to~$G'$ resulting in that the number of red and blue vertices being equal.
	Finally, we set~$k' \coloneqq 2k$ and $\delta = 0$.\footnote{Indeed, the construction works for any $\delta \ge 0$.}
\end{construction}

Let us prove the correctness of the above reduction.

\begin{lemma}
	\label{lem:hard-1}
	Given an instance~$I = (G, k)$ of \CE{}, \cref{constr:hard-1} returns an instance~$I' = (G', k', \delta)$ of \FCE{}
	such that~$I$ is a \yes-instance if and only if~$I'$ is a \yes-instance.
	Moreover, whenever~$I'$ is a \yes-instance, there exists a solution which only deletes edges.
\end{lemma}

\begin{proof}
	Assume first that~$I$ is a \yes-instance.
	By \cref{prop:ce-hard}(i) and (ii), we may assume that any solution for~$I$ requires exactly~$k$ edge deletions.
	Then, deleting the corresponding~$k$ edges in~$G'$ and also one arbitrary edge of each red~$P_3$ in~$G'$ clearly yields a solution~$S'$ of size~$2k=k'$ with~$\diff(S')=0\le\delta$ (as~$G'$ contains the same number of red and blue vertices).
	Note that~$S'$ contains only edge deletions.

	Conversely, let~$(G', k', \delta)$ be a \yes-instance.
	Note that every solution modifies at least one edge of each of the~$k$ red~$P_3$s in~$G'$.
	As the~$P_3$s are all pairwise vertex-disjoint,
	we may assume without loss of generality that every such modification is a deletion.
	Hence, at most~$k$ edge deletions are performed to transform the copy of~$G$ in~$G'$ into a cluster graph.
\end{proof}

We now provide the construction for \cref{thm:fce-hard}(ii).
Again, we will later see that the construction also proves \NP-hardness for the \prob{Deletion} variant.

\begin{construction}[for \cref{thm:fce-hard}(ii)]
	\label{constr:hard-2}
	Let~$I = (G, k)$ be an instance of \CE{}.
	We may assume that~$I$ has the properties listed in \cref{prop:ce-hard}.
	We construct an instance~$I' = (G', k', \delta)$ of \FCE{} as follows.
	The graph~$G'$ contains a blue copy of~$G$ as well as one red vertex~$r$ which is adjacent to an arbitrary vertex~$x$ of degree six from~$G$ (this exists due to \cref{prop:ce-hard}(iv)).
	We further add~$2k - \abs{V(G)} + 1$ isolated blue vertices such that overall~$G'$ contains~$2k+1$ blue vertices.
	Note that~$2k - \abs{V(G)} + 1 > 0$ due to \cref{prop:ce-hard}(iii).
	Finally, we set~$k' \coloneqq k+1$ and $\delta = 0$.\footnote{Just as \cref{constr:hard-1}, this construction works with any $\delta \ge 0$.}
\end{construction}

Again, let us prove the reduction to be correct.

\begin{lemma}
	\label{lem:hard-2}
	Given an instance~$I = (G, k)$ of \CE{}, \cref{constr:hard-2} returns an instance~$I' = (G', k', \delta)$ of \FCE{}
	such that~$I$ is a \yes-instance if and only if~$I'$ is a \yes-instance.
	Moreover, whenever~$I'$ is a \yes-instance, there exists a solution which only deletes edges.
\end{lemma}

\begin{proof}
	If~$(G,k)$ is a \yes-instance with solution~$S$, then $S'\coloneqq S\cup\{\{r,x\}\}$ yields a solution for~$G'$ of size~$k+1=k'$ with~$\diff(S')=|\frac{1}{1} - \frac{2k+1}{2k+1}|=0\le\delta$.
	As we may assume that~$S$ contains only edge deletions, we may assume the same for~$S'$.

	Conversely, suppose that~$(G',k',\delta)$ is a \yes-instance with solution~$S'$.
	By \cref{prop:ce-hard}(ii), any modification set that transforms~$G$ into a cluster graph contains at least~$k$ edge deletions, and each vertex in~$G$ is incident to at most~$4$ deletions.
	Hence, the same holds true for~$S'$ restricted to~$V(G') \setminus \{r\}$. 
	In other words, there are at least~$k = k' - 1$ edge deletions in~$S'$ that are not incident to~$r$ and each vertex in~$V(G)$ is incident to at most~$4$ of them.
	We claim that~$\{r, x\} \in S'$.
	Suppose not.
	Then the $P_3$s induced by~$r$, $x$, and any neighbor~$v \ne r$ of~$x$ must be resolved by either deleting~$\{v, x\}$ or by adding~$\{v, r\}$.
	If we resolve more than one of these~$P_3$s by adding the edges~$\{v, r\}$, then the remaining budget is less than~$k$ and thus does not suffice to transform the remaining graph into a cluster graph.
	So we have to resolve at least five of the~$P_3$s by deleting the corresponding edge~$\{v, x\}$.
	This however contradicts the fact that every vertex in~$V(G)$ is incident to at most~$4$ modifications within~$G$.
	Therefore, $\{r, x\} \in S'$, and the remaining $k$ modifications in~$S'$ are within~$G$; hence~$(G, k)$ is a \yes-instance.
\end{proof}

\cref{thm:fce-hard} now follows from \cref{prop:ce-hard,lem:hard-1,lem:hard-2}, together with the following observation.

\begin{observation}
	\cref{constr:hard-1,constr:hard-2} run in polynomial time.
	Moreover, for any instance~$I' = (G', k', \delta)$ returned by either construction, we have
	$k' \in \bigO(k)$, $\abs{V(G')} \in \bigO(\max\{\abs{V(G)}, k\})$, and~$\abs{E(G')} \in \bigO(\abs{E(G)} + k)$.
\end{observation}

We remark that \cref{thm:fce-hard}~(i) also holds if the maximum degree is six and the maximum number of edge modifications (or deletions) incident to each vertex is at most four.
These are immediate consequences of properties~(iv) and~(v) of \cref{prop:ce-hard}.

Surprisingly, \CC{}, which is trivially solvable in polynomial time, becomes \NP-hard when enforcing fairness.

\begin{theorem}\label{thm:FCC-NP-hard}
	\FCC{} is \NP-hard for every constant~$\delta\ge 0$.
	This also holds if only mono-colored edge insertions are allowed.
\end{theorem}

The proof is based on a polynomial-time reduction from the following problem, which
we will later prove to be \NP-hard in \cref{sec:transform} (\cref{thm:cta-nph}).

\probDef{\CTA}
{A cluster graph~$G$ and an integer~$k \in \Nzero$.}
{Can~$G$ be transformed into another cluster graph by adding exactly~$k$ edges?}

\begin{construction}[for \cref{thm:FCC-NP-hard}]
	\label{constr:FCC-NP-hard}
	Let~$I = (G, k)$ be an instance of \CTA{} with~$n \coloneqq \abs{V(G)}$ and~$m \coloneqq \abs{E(G)}$ and assume without loss of generality that~$k \le \binom{n}{2} - m$, otherwise~$I$ is a trivial \no-instance.
	Let~$\delta \ge 0$ be an arbitrary constant.
	Choosing a sufficiently large instance~$I$, we may assume that~$\delta \le 2(m-n-1)/n$.
	We construct an instance~$I' = (G',k',\delta)$ with~$k'\coloneqq 2k+\lfloor\frac{\delta n}{2}\rfloor$ as follows.
	The graph~$G'$ contains a copy of~$G$ where every vertex is colored blue
	together with~$n$ red vertices which form an arbitrary connected graph with~$x \coloneqq \binom{n}{2}-k - \lfloor\frac{\delta n}{2}\rfloor$ edges.
	Note that this is possible as
	\[
		x \ge \ctwo{n} - \left(\ctwo{n}-m\right) - \left\lfloor \frac{ n \cdot 2(m-n-1)/n }{2} \right\rfloor \ge n-1. \qedhere
	\]
\end{construction}

Let us prove the correctness of the reduction. 

\begin{lemma}
	\label{lem:FCC-NP-hard}
	Given an instance~$I = (G, k)$ of \CTA{},
	\cref{constr:FCC-NP-hard} returns an instance~$I' = (G', k', \delta)$ of \FCC{} such that
	$I$ is a \yes-instance if and only if~$I'$ is a \yes-instance.
\end{lemma}
\begin{proof}
	Assume that~$(G,k)$ is a \yes-instance.
	Then, adding the corresponding~$k$ edges to the blue copy of~$G$ in~$G'$ and the~$k+\lfloor\frac{\delta n}{2}\rfloor$ missing edges to the red subgraph yields a cluster graph.
	This set~$S'$ of added edges satisfies
	\[
		\diff(S') = \frac{2(k + \lfloor\frac{\delta n}{2}\rfloor)}{n} - \frac{2k}{n} \le \frac{2k + \delta n - 2k}{n} = \delta.
	\]

	Conversely, let~$(G',k',\delta)$ be a \yes-instance.
	By our problem definition, our corresponding solution~$S'$ of size~$k'$ contains the~$k_r \coloneqq k + \lfloor \frac{\delta n}{2} \rfloor$ missing edges of the red subgraph of~$G'$.
	Let~$k_b$ be the number of edges between blue vertices and~$k^*$ be the number of edges between a blue and a red vertex in~$S'$ and note that~$k_b + k^* = k$.
	As we have~$k_r$ ($k_b$) edges with two red (blue) endpoints and~$k^*$ edges with one endpoint of each color,
	we have
	\begin{align*}
		\diff(S') &= \frac{2k_r+k^*}{n} - \frac{2k_b + k^*}{n} = \frac{2(k_r - k_b)}{n} = \frac{2(k+ \lfloor \frac{\delta n}{2} \rfloor - k_b)}{n} \\
			  &\ge \frac{2(k-k_b) + \delta n - 1)}{n}  = \delta + \frac{2 k^* - 1}{n}.
	\end{align*}
	As~$\diff(S') \le \delta$ we have that~$k^* \le 0$.
	So~$k_b = k$ and~$S'$ contains $k$~edges in the blue copy of~$G$ in~$G'$; thus~$(G, k)$ is a \yes-instance.
\end{proof}

As \cref{constr:FCC-NP-hard} is clearly computable in polynomial time, \cref{thm:FCC-NP-hard} follows immediately from \cref{lem:FCC-NP-hard,thm:cta-nph}.

We observe from the intractability results so far that the hardness of \FCE{} is rooted in finding the right mono-colored edge modifications.
Indeed, we can show that, if only $\mu$ mono-colored edge modifications are allowed for constant~$\mu$, then there is a randomized polynomial-time algorithm.
We will prove that this can be done by guessing the~$\mu_R$ and~$\mu_B$ modifications between red and between blue endpoints before reducing to the \prob{Budgeted Matching} problem:
Given a graph~$H$ with edge weights~$w \colon E(H) \to \QQ^+$, edge cost~$c \colon E(H) \to \QQ^+$, and weight and cost bounds~$W, C \in \QQ^+$,
the problem asks whether there is a matching~$M \subseteq E(H)$ with~$w(M) \coloneqq \sum_{e \in M} w(e) \ge W$ and~$c(M) \coloneqq \sum_{e \in M} c(e) \le C$.
Recall that an edge set~$M \subseteq E(H)$ is a \emph{matching} if no two edges in~$M$ share an endpoint.
\citet[Lem.~8]{berger2011budgeted} have shown that, if all edge weights and costs and the budget are polynomially bounded in the size of the input graph, \prob{Budgeted Matching} can be reduced in polynomial time to the \prob{Exact Perfect Matching} problem, in which, given an~$n$-vertex graph in which some edges are red and an integer~$k \in \Nzero$, the task is to decide whether there exists a matching of size~$n/2$ that contains \emph{exactly} $k$ red edges.
For \prob{Exact Perfect Matching}, there is a randomized algorithm without false positives and error probability at most a given~$\eps > 0$  with running time~$n^{\bigO(1)} \log 1/\eps$~\cite{MulmuleyVV87}. (Notably, it is unknown whether there exists a deterministic polynomial-time algorithm for the problem.)
Due to the reduction by \citet{berger2011budgeted}, \prob{Budgeted Matching} can be solved by a randomized algorithm with asymptotically the same running time bound and error probability as the one for \prob{Exact Perfect Matching}.

\begin{theorem}
	\label{thm:mu=0}
	Let~$\eps > 0$.
	Then there is a randomized algorithm without false positives and error probability at most~$\eps$
	that solves \FCE{} in~$n^{\bigO(\mu)}\log 1/\eps$ time, where~$n$ is the number of vertices and~$\mu$ is the number of allowed mono-colored modifications.
\end{theorem}

\begin{proof}
	Let~$(G, k, \delta)$ with~$V(G) = R \uplus B$ be an instance of \FCE{} and assume without loss of generality that~$\abs{R} \ge \abs{B}$.
	We first guess the numbers~$\mu_R$ and~$\mu_B$ with~$\mu_R + \mu_B \le \mu$
	and the mono-colored modification sets~$S_R \subseteq \binom{R}{2}$ and~$S_B \subseteq \binom{B}{2}$ of size~$\mu_R$ and~$\mu_B$.
	Let~$G^*$ be the graph obtained after applying the modifications in~$S_R$ and~$S_B$ to~$G$.
	Note that~$G^*[R]$ and~$G^*[B]$ must be cluster graphs as we can only do bi-colored edge modifications from here on.
	Now, for any hypothetical bi-colored edge modification set~$S' \subseteq \binom{V(G)}{2} \setminus (\binom{R}{2} \cup \binom{B}{2})$, we require
	\[
		\diff(S_R \cup S_B \cup S') = \left\lvert \frac{2\mu_R + \abs{S'}}{\abs{R}} - \frac{2\mu_B + \abs{S'}}{\abs{B}} \right\rvert \le \delta,
	\]
	which is equivalent to requiring
	\[
		-\delta \le \frac{2\mu_R + \abs{S'}}{\abs{R}} - \frac{2\mu_B + \abs{S'}}{\abs{B}} \le \delta.
	\]
	If~$\abs{R} = \abs{B}$, then adding bi-colored edges will have no effect on~$\diff(S_R \cup S_B \cup S')$; thus we assume that~$\abs{R} > \abs{B}$.
	Then, adding bi-colored edges will increase the average number of edits incident to~$B$ more than those incident to~$R$.
	The above inequalities yield the following lower and upper bound on~$\abs{S'}$:
	\begin{align*}
		\alpha' \coloneqq
		\frac{-\delta - (2\mu_R/\abs{R} - 2\mu_B/\abs{B})}{1/\abs{R} - 1/\abs{B}}
		\le \abs{S'} \le
		\frac{\delta - (2\mu_R/\abs{R} - 2\mu_B/\abs{B})}{1/\abs{R} - 1/\abs{B}} \eqqcolon \beta'.
	\end{align*}
	Note that~$\alpha'$ and~$\beta'$ may be negative and larger than~$k-\mu$;
	thus we may look for a bi-colored edge modification set of size at least $\alpha \coloneqq \max\{0, \alpha'\}$ and at most~$\beta \coloneqq \min\{k-\mu, \beta'\}$.

	Let~$R_1,\ldots,R_r$ and $B_1, \dots, B_b$ be the vertex sets of the clusters in~$G^*[R]$ and~$G^*[B]$, respectively.
	Since~$S'$ shall only contain bi-colored edges,
	a solution can never merge two blue or two red clusters into one.
	Thus, any solution either isolates a cluster, or merges it with exactly one cluster of the other color.
	This can be modeled as a matching in a complete bipartite graph~$H$ with vertices~$u_1,\ldots,u_r$ on one side and~$v_1,\ldots,v_b$ on the other side, where a matching edge indicates which clusters are merged.
	Clearly, every cluster editing solution for~$G^*$ with only bi-colored edits corresponds to a matching and vice versa.
	Let~$E'\subseteq E(G)$ be the edges between~$R$ and~$B$ and let~$E_{ij}\subseteq E'$ denote the edges between $R_i$ and~$B_j$.
	For a given matching~$M$ in~$H$, a solution
	must remove all edges in~$E'$ except for those in~$E_{ij}$ corresponding to a matching edge~$\{u_i, v_j\} \in M$.
	Further, for every matching edge~$\{u_i, v_j\}$, we must add all~$\abs{R_i}\abs{B_j} - \abs{E_{ij}}$ missing edges.
	Hence, the size of a bi-colored modification set~$S'$ corresponding to~$M$ is
	\[|E'|-\;\sum_{\mathclap{\{u_i,v_j\}\in M}}\;|E_{ij}|+ \;\sum_{\mathclap{\{u_i,v_j\}\in M}}\;(|R_i||B_j|-|E_{ij}|) = \abs{E'} - \;\sum_{\mathclap{\{u_i, v_j\} \in M}} (2\abs{E_{ij}} - \abs{R_i}\abs{B_j}). \]

	Define~$w \colon E(H) \to \QQ^+$ with~$w(\{u_i, v_j\}) \coloneqq (2\abs{E_{ij}} - \abs{R_i}\abs{B_j})$.
	Then, a matching~$M$ of weight~$\abs{E'} - \beta \le w(M) \le \abs{E'} - \alpha$ corresponds to a bi-colored modification set~$S'$ such that~$S_R \cup S_B \cup S'$ transforms~$G$ into a cluster graph and~$\diff(S_R \cup S_B \cup S') \le \delta$.
	To this end, we solve an instance for \prob{Budgeted Matching} with cost function~$c \equiv w$ and budgets~$W \coloneqq \abs{E'} - \beta$ and~$C \coloneqq \abs{E'} - \alpha$
	using the reduction~\cite{berger2011budgeted} and randomized algorithm~\cite{MulmuleyVV87} mentioned above.
	The algorithm returns the desired matching~$M$ with probability at least~$1-\eps$ if it exists, and reports \no{} otherwise.
	In the former case, we return the modification set~$S_R \cup S_B \cup S'$ and thus correctly report \yes{} with probability at least~$1-\eps$.
	If the algorithm reports \no{} for every possible guessed mono-colored edge modification set~$S_R \cup S_B$, then we report that there is no modification-fair modification set of size at most~$k$.
	Let~$n \coloneqq \abs{V(G)}$.
	As there are~$\binom{n}{2}^\mu \le n^{2\mu}$ guesses, for each of which we solve an instance of \prob{Budgeted Matching} in~$\abs{V(H)}^{\bigO(1)}\log1/\eps \subseteq n^{\bigO(1)}\log 1/\eps$ time, the running time follows.
\end{proof}

We leave open whether or not \FCE{} is fixed-parameter tractable when parameterized by the number~$\mu$ of mono-colored edge modifications.
However, for the larger parameter~$k$, the number of edge modifications, we are able to prove fixed-parameter tractability --- we will prove this next.
Our approach is as follows.
We first run the well-known 
$P_3$-branching algorithm~\cite{Cai96} to enumerate cluster graphs.
As the resulting solution need not be modification-fair, we may need to do further edge modifications.
For this, we first apply polynomial-time data reduction rules which shrink the graph size to a polynomial in~$k$, and then brute-force on the reduced graph.

\begin{theorem}\label{thm:fpt}
	\FCE{} can be solved in $2^{\bigO(k \log k)} \cdot (n+m)$ time on $n$-vertex, $m$-edge graphs.
\end{theorem}

\begin{proof}
	Let~$(G, k, \delta)$ be an instance of \FCE{} with~$V(G) = R \uplus B$.
	We first apply the standard $P_3$-branching algorithm for \CE{} to enumerate all minimal cluster edge modification sets~$S$ of size at most~$k$ in~$\bigO(3^k (n+m))$ time~\cite{Cai96}.
	For each~$S$, we check whether $\diff(S)\le \delta$.
        If not, then we try to extend~$S$ to a fair edge modification set.
	Clearly, each fair edge modification set of size at most~$k$ contains at least one of the enumerated edge modification sets.
        Note that in order to check later that our modification set is fair, we store the original numbers~$|B|$ and $|R|$ of blue and red vertices in~$G$.
        
        For each~$S$, we first apply the following three data reduction rules to the cluster graph~$G'$ obtained from~$S$. 
	\begin{enumerate}
		\item If there is a clique with more than~$k+1$ vertices, then delete it.
		\item If there are more than~$2k$ isolated vertices of the same color which have not been touched by~$S$, then delete one of them.
		\item Let~$2 \le s \le k+1$ and~$0\le t \le s$.
			If there are more than~$k$ cliques with~$s$ vertices, $t$ of which are blue, and none of them are touched by~$S$, then delete one of them.
	\end{enumerate}
	Note that we keep all cliques with at most~$k+1$ in~$G'$ which contain an endpoint of an edge in~$S$. Clearly, there are at most~$2|S|$ such cliques.

	For the correctness, note that modifying a clique with~$\ell\ge 2$ vertices requires at least~$\ell-1$ edge modifications.
	Hence, Rule~1 is correct.
	Clearly, $k$ edge modifications can touch at most~$2k$ vertices of any color, so Rule~2 is correct.
	Rule~3 is correct as we cannot touch more than~$k$ cliques of size at least two.

	For exhaustive application of the data reduction rules, we count the number of cliques with the same numbers of blue and red vertices.
	As we added at most~$k$ edges to obtain~$G'$, we can apply the rules in~$\bigO(n+m+k)$ time.
	After exhaustive application, the remaining graph contains~$\bigO(k^2)$ vertices contained in cliques touched by~$S$ and~$\bigO(k^3)$ vertices not touched by~$S$.

	Let~$W\subseteq V(G)$ be the vertices remaining after exhaustive application of the above data reduction rules.
	We now try all possible extensions~$S' \subseteq \binom{W}{2}\setminus S$ of size at most~$k-|S|$ and check whether the set~$S^* \coloneqq S \cup S'$ transforms~$G$ into a cluster graph and is fair, that is, $\diff(S^*) \le \delta$.
	There are~$\bigO(k^{6k})$ such extensions; the checking can be done in~$\bigO(m+n+k)$ time each.
	The overall running time thus is~$2^{\bigO(k\log k)}\cdot(m+n)$.
\end{proof}

Seeing this approach, one may ask why it cannot be adapted to prove fixed-parameter tractability for the number~$\mu$ of mono-colored edge edits.
Of course, we can use the standard branching algorithm to enumerate all minimal solutions for~$G[R]$ and~$G[B]$ in~$\bigO(3^\mu (n+m))$ time.
However, we cannot apply the three data reduction rules, as we can differentiate between the clusters in~$G[R]$ and~$G[B]$ due to their incident bi-colored edges.
Hence, it is not clear which clusters we can safely discard.

\section{Transforming Cluster Graphs}\label{sec:transform}
\appendixsection{sec:transform}

This section is devoted to proving the NP-hardness of the above introduced \CTA{}.
Recall that in this problem we are given a cluster graph~$G$ and an integer~$k \in \Nzero$,
and we are asked to decide whether $G$ can be transformed into \emph{another} cluster graph by adding exactly~$k$ edges.

\begin{theorem}
	\label{thm:cta-nph}
	\CTA{} is \NP-hard.
\end{theorem}

We devise a polynomial-time reduction from the \textsc{Numerical 3D Matching} problem introduced and proven to be strongly \NP-hard by \citet{GJ75}.
Herein, given positive integers~$t$, $a_1, \dots, a_n$, $b_1, \dots, b_n$, and $c_1, \dots, c_n$, one is asked whether there are bijections
$\alpha, \beta, \gamma \colon \oneto{n} \to \oneto{n}$ such that $a_{\alpha(i)} + b_{\beta(i)} + c_{\gamma(i)}=t$ holds for each~$i \in \oneto{n}$.

On a high level, our reduction works as follows.
We add a \emph{small} clique for every~$a_i$, a \emph{medium-sized} clique for every~$b_i$, and a \emph{large} clique for every~$c_i$.
Throughout this section, we will refer to the number of vertices in a clique as its \emph{size}.
By appropriate choice of our solution size~$k$, we can ensure that every clique in the resulting cluster graph --- our so-called \emph{solution graph}~$G'$ with vertex set~$V(G') = V(G)$ and edge set~$E(G') = E(G) \cup S$ --- is the result of merging one small, one medium, and one large clique.
We finally show that if each cluster consists of cliques corresponding to elements~$a_i$, $b_j$, and~$c_\ell$ such that their sum is equal to the target~$t$, then the number of required edge additions is minimized.
That is, if there is a cluster that does not hit this target, then the resulting solution adds more than~$k$ edges.

\begin{construction}[for \cref{thm:cta-nph}]
	\label{constr:cta-nph}
	Let~$I=(t, a_1, \dots, a_n, b_1, \dots, b_n,\allowbreak c_1, \dots, c_n)$, $n\ge 3$, be an instance of \textsc{Numerical 3D Matching}.
	As \textsc{Numerical 3D Matching} is strongly \NP-hard, we may assume that for all~$i \in \oneto{n}$, $a_i, b_i, c_i \le n^d$ for some constant~$d > 0$.
	We further assume that~$t > a_i, b_i, c_i$ for all~$i \in \oneto{n}$ and that~$\sum_{i=1}^n (a_i + b_i + c_i) = n \cdot t$, as otherwise $I$ is a trivial \no-instance.

	We construct an instance~$I'=(G, k)$ of \CTA{} as follows.
	Let~$A \coloneqq n^{2d}$, let~$B \coloneqq n^{3d}$, and let~$C \coloneqq n^{7d}$.
	For~$i \in \oneto{n}$, we set~$a'_i \coloneqq a_i+A$, $b'_i \coloneqq b_i+B$, $c'_i \coloneqq c_i+C$, and add three cliques of size~$a'_i$, $b'_i$, and~$c'_i$, respectively, to~$G$.
	We refer to these cliques by their size~$a'_i$, $b'_i$, $c'_i$ and call them \emph{small}, \emph{medium-sized}, and \emph{large}, respectively.
	For more convenient notation, let~$t' \coloneqq t + A + B + C$.
	Finally, set
        \[k \coloneqq n\smallbinom{t'}{2}-|E(G)|= n\smallbinom{t'}{2}- \sum_{i=1}^n\Big(\smallbinom{a'_i}{2}+\smallbinom{b'_i}{2}+\smallbinom{c'_i}{2}\Big). \qedhere\]
\end{construction}

Proving the forward direction of our reduction is straightforward.

\begin{lemma}
	\label{obs:cta-forward}
	If \cref{constr:cta-nph} is given a \yes-instance~$I$ of \prob{Numerical 3D Matching},
	then it returns a \yes-instance~$I'$ of \CTA{}.
\end{lemma}

\begin{proof}
	Let~$\alpha, \beta, \gamma$ be a solution for instance~$I$.
	Creating~$n$ clusters by merging the cliques~$a'_{\alpha(i)}$, $b'_{\beta(i)}$, $c'_{\gamma(i)}$ for each~$i \in \oneto{n}$ yields a solution graph~$G'$ with
	\begin{align*}
		|E(G')| &= \sum_{i=1}^n \ctwo{a'_{\alpha(i)} + b'_{\beta(i)} + c'_{\gamma(i)}} = \sum_{i=1}^n \ctwo{t+A+B+C} = \sum_{i=1}^n \ctwo{t'}
	\end{align*}
	edges, created by adding~$|E(G')|-|E(G)| = k$ edges.
\end{proof}

The backward direction is more involved.
In the following, let~$I' = (G, k)$ be an instance of \CTA{} obtained from applying \cref{constr:cta-nph} on an instance~$I$ of \prob{Numerical 3D Matching}.
We will first provide a lower and an upper bound on~$k$.
Then, step by step, we will prove that every solution of our constructed instance~$I'$ transforms our graph into a cluster graph with~$n$ cliques, each containing exactly one small, one medium-sized, and one large clique.

\begin{lemma}
	\label{obs:cta-nph-k}
	In the constructed instance~$I'$ we have
	$n(AC+BC) \le k \le 2nBC$.
\end{lemma}
\begin{proof}
	It is easy to verify that for~$x_1, \dots, x_n \in \NN$,
	\begin{align*}
		\binom{\sum_{i=1}^n x_i}{2} &= \frac{1}{2} \bigg( \Big( \sum_{i=1}^n x_i\Big)^2 - \sum_{i=1}^n x_i \bigg) = \frac{1}{2} \bigg( \sum_{i=1}^n x_i^2 + 2 \sum_{\mathclap{i<j \in \oneto{n}}} x_i x_j - \sum_{i=1}^n x_i \bigg)\\
					    &= \sum_{i=1}^n \ctwo{x_i} + \sum_{i=1}^n\sum_{j=i+1}^n x_ix_j.
	\end{align*}
	Hence, we can reformulate~$k$ as follows:
	\begin{align*}
		k &= \sum_{i=1}^n \Big(\ctwo{t'} - \ctwo{a'_i} - \ctwo{b'_i} - \ctwo{c'_i}\Big)\\
		  &= \sum_{i=1}^n \Big(\ctwo{t} + \ctwo{A} + \ctwo{B} + \ctwo{C} + AB + BC + AC + t(A + B + C)\\
		  &\phantom{=} -\ctwo{A} - \ctwo{a_i} - A a_i - \ctwo{B} - \ctwo{b_i} - B b_i - \ctwo{C} - \ctwo{c_i} - C c_i\Big) \\
		  &= \sum_{i=1}^n \Big(\ctwo{t} - \ctwo{a_i} - \ctwo{b_i} - \ctwo{c_i} + AB + BC + AC\\
		  &\phantom{=}+ (t-a_i)A + (t-b_i)B + (t-c_i)C\Big).
		  \addtocounter{equation}{1}\tag{\theequation}\label{eq:cta-nph-k}
	\end{align*}
	As~$1 \le a_i, b_i, c_i \le n^d$, we have~$\binom{a_i}{2} + \binom{b_i}{2} + \binom{c_i}{2} \le 3n^{2d}\le n^{5d} = AB$.
	Thus, $k \ge n (AC+BC)$.
	For the upper bound, observe that~$t \le 3n^d$.
	Plugging this into \eqref{eq:cta-nph-k}, we have
	\begin{align*}
		k &\le \sum_{i=1}^n \Big(\ctwo{3n^d} + AB + BC + AC + 3n^d(A+B+C)\Big)\\
		     &\le nBC + n\big(9n^{2d} + n^{5d} + n^{9d} + 3n^{3d} + 3n^{4d} + 3n^{8d}\big).
	\end{align*}
	As the right summand is a polynomial in~$n$ with degree~$9d+1$ and coefficients in~$\bigO(1)$, it is at most~$nBC$ for sufficiently large~$n$.
	Thus, the upper bound follows.
\end{proof}

\cref{obs:cta-nph-k} implies that we cannot merge two large cliques.
\begin{lemma}
	\label{lem:cta-nph-cc}
	If~$I'$ is a \yes-instance, then no solution merges two large cliques.
\end{lemma}

\begin{proof}
	Merging two large cliques adds at least~$C^2 = n^{14d} > 2n^{10d+1} = 2nBC > k$ edges, see \cref{obs:cta-nph-k}.
\end{proof}

At the same time, we cannot reach our budget unless we merge every medium-sized clique with a large one.

\begin{lemma}
	\label{lem:cta-nph-bc}
	If~$I'$ is a \yes-instance, then in any solution graph, for every~$i\in\oneto{n}$, the clique~$b'_i$ is merged with exactly one clique~$c'_{\varphi(i)}$, where~$\varphi(i) \in \oneto{n}$.
\end{lemma}
\begin{proof}
	By \cref{lem:cta-nph-cc}, we can merge any small clique and any medium-sized clique with at most one large clique.
	Note that every medium-sized clique that is merged with a large clique contributes more than $BC$ edge additions to our budget.
	Assume towards a contradiction that
        there is one medium-sized clique that is not merged with a large clique.
	Then, the maximum number of edge additions is achieved by
	merging the other~$n-1$ medium-sized cliques and all~$n$ small cliques with the largest clique, leaving the one remaining medium-sized clique and~$n-1$ large cliques untouched.
	The number of edges added by these merges is at most
	\begin{align*}
		&(n-1)(B+n^d)(C+n^d) + \ctwo{n-1}(B+n^d)^2\\
		&+ n(n-1)(A+n^d)(B+n^d) + n(A+n^d)(C+n^d)+ \ctwo{n}(A+n^d)^2\\
		=&~ (n-1)BC + R,
	\end{align*}
	where~$R$ is a polynomial in~$n$ with degree at most~$9d+1$ (due to the summand $nAC = n^{9d+1}$) and coefficients in~$\bigO(1)$.
	Thus, for sufficiently large~$n$, we have~$R < n^{10d} = BC$, and the overall number of added edges is less than~$nBC$, which by \cref{obs:cta-nph-k} is a contradiction.
\end{proof}

Now we know that a significant part of the budget is spent on merging medium-sized cliques with large cliques.
Nevertheless, we cannot meet our budget unless we spend the remaining budget on merging small cliques with a large clique as well.

\begin{lemma}
	\label{lem:cta-nph-ac}
	If~$I'$ is a \yes-instance, then in any solution graph, for every~$i\in\oneto{n}$, the clique~$a'_i$ is merged with exactly one clique~$c'_{\chi(i)}$, where~$\chi(i) \in \oneto{n}$.
\end{lemma}
\begin{proof}
	By \cref{lem:cta-nph-cc} and \cref{lem:cta-nph-bc},
	we can merge merge any small clique with at most one medium-sized clique and at most one large clique.
	Assume towards a contradiction that there is one small clique that is not merged with a large clique.
	Then, the maximum number of edge additions is achieved by
	merging~$n-1$ small cliques
	with all~$n$ medium-sized cliques
	and one large clique,
	and leaving the remaining small clique and $n-1$ large cliques untouched.
	The number of edge additions provided by this is at most
	\begin{align*}
		&\ctwo{n-1}(A+n^d)^2
		+ n(n-1)(A+n^d)(B+n^d)\\
		&+ (n-1)(A+n^d)(C+n^d)
		+ \ctwo{n}(B+n^d)^2
		+ n(B+n^d)(C+n^d)\\
		={}& nBC + (n-1)AC + R,
	\end{align*}
	where~$R$ is a polynomial in~$n$ with degree at most~$8d+1$ (due to the summand $n\cdot C \cdot n^d = n^{8d+1}$) and coefficients in~$\bigO(1)$.
	Thus, for sufficiently large~$n$, we have~$R < n^{9d} = AC$, and the overall number of added edges is less than~$nBC+nAC$, which by \cref{obs:cta-nph-k} is a contradiction.
\end{proof}

Combining \cref{lem:cta-nph-cc,lem:cta-nph-bc,lem:cta-nph-ac}, we obtain the following.

\begin{lemma}
	\label{lem:cta-nph-n}
	If~$I'$ is a \yes-instance, then every solution graph consists of exactly~$n$ cliques.
\end{lemma}

\begin{proof}
	By \cref{lem:cta-nph-cc}, no two large cliques can be merged, that is, the solution graph contains at least~$n$ cliques.
	By \cref{lem:cta-nph-bc,lem:cta-nph-ac}, every medium-sized clique and every small clique must be merged with exactly one large clique, which implies that the solution graph contains at most~$n$ cliques.
\end{proof}

With \cref{lem:cta-nph-n} at hand, we can show that the budget~$k$ is exactly met if and only if each resulting clique contains $t'$~vertices.

\begin{lemma}
	\label{lem:cta-nph-t}
	If~$I'$ is a \yes-instance, then every clique in a solution graph~$G'$ contains~$t'$ vertices.
\end{lemma}
\begin{proof}
	By \cref{lem:cta-nph-n}, $G'$ consists of $n$ cliques.
	Let their sizes be~$s_1, s_2, \dots, s_n$.
	Then their sum is~$\sum_{i=1}^n s_i = nt'$; otherwise the \textsc{Numerical 3D Matching} instance~$I$ is a \no-instance.
	Next, note that
	\begin{align}
		\label{eq:n}
		\abs{E(G')} = \sum_{i=1}^n \ctwo{s_i} &= \frac{1}{2}\sum_{i=1}^n s_i^2 - \frac{1}{2}\sum_{i=1}^n s_i = \frac{1}{2}\sum_{i=1}^n s_i^2 - \frac{1}{2} nt'.
	\end{align}
	As~$\abs{E(G')} = \abs{E(G)} + k = n\binom{t'}{2} = \frac{1}{2} (n{t'^2} - nt')$,
	we have~$\sum_{i=1}^n s_i^2 = n{t'}^2$.
	By the Cauchy-Schwarz inequality, we have
	\begin{align}
		\label{eq:cs}
		n \cdot \sum_{i=1}^n s_i^2 &= \bigg(\sum_{i=1}^n 1^2\bigg)\bigg(\sum_{i=1}^n s_i^2\bigg) \ge \bigg(\sum_{i=1}^n 1 \cdot s_i\bigg)^2 = (nt')^2 = n \cdot n{t'}^2,
	\end{align}
        that is, we have $\sum_{i=1}^ns_i^2 \ge n{t'}^2$, and the two sides are equal only if~$s_1 = s_2 = \dots = s_n = t'$.
\end{proof}

This allows us to prove the desired property of any solution for $I'$.

\begin{lemma}
	\label{lem:cta-nph-abc}
	If~$I'$ is a \yes-instance, then every clique in a solution graph~$G'$ consists of a small, a medium-sized, and a large clique.
\end{lemma}
\begin{proof}
	By \cref{lem:cta-nph-t}, every clique in a solution graph contains exactly~$t'$ vertices.
	As~$t' = t + A + B + C = t + n^{2d} + n^{3d} + n^{7d}$ and~$t \le 3n^d$, every clique in the solution graph must consist of a small, a medium-sized, and a large clique;
	otherwise the clique cannot consist of exactly~$t'$ vertices.
\end{proof}

Now, proving the backward direction of our reduction is straightforward.

\begin{lemma}
	\label{lem:cta-backward}
	Let~$I$ is an instance of \prob{Numerical 3D Matching}
	and let~$I'$ be the instance of \CTA{} obtained by applying \cref{constr:cta-nph} on~$I$.
	If~$I'$ is a \yes-instance, then so is~$I$.
\end{lemma}

\begin{proof}
	Let~$S$ be a solution for instance~$I'$ and let~$G' \coloneqq (V(G), E(G) \cup S)$ be the corresponding solution graph.
	By \cref{lem:cta-nph-n}, $G'$ consists of~$n$ clusters of size~$s_1, s_2, \dots, s_n$.
	By \cref{lem:cta-nph-abc}, there are~$\alpha, \beta, \gamma \colon \oneto{n} \to \oneto{n}$ such that, for every~$i \in \oneto{n}$, we have
	\[ s_i = a'_{\alpha(i)} + b'_{\beta(i)} + c'_{\gamma(i)} = a_{\alpha(i)} + b_{\beta(i)} + c_{\gamma(i)} + A + B + C = t + A + B + C. \]
	Hence, $\alpha, \beta, \gamma$ is a solution for instance~$I$.
\end{proof}

We now have everything at hand to prove \cref{thm:cta-nph}.

\begin{proof}[Proof of \cref{thm:cta-nph}]
	We use \cref{constr:cta-nph} to build an instance~$I'$ of \CTA{} from a given instance~$I$ of \textsc{Numerical 3D Matching}.
	Clearly, $I'$ can be computed in polynomial time.
	By \cref{obs:cta-forward} and \cref{lem:cta-backward}, $I$ is a \yes-instance if and only if~$I'$ is a \yes-instance.
\end{proof}

\section{Empirical Insights into the Price of Fairness}
\label{sec:experiments}

\newcommand{\ceexperimentpath}{experiment-data}
We now study our model of modification fairness empirically, the focus being the \emph{price of modification fairness}:
How fair are colorblind solutions, and how much do we have to pay (in solution cost and in computation time) in comparison with colorblind solutions, to obtain a (sufficiently) fair solution?
In the spirit of \citet{BBK11} who studied classic \CE{}, we refrain from using our algorithm proving fixed-parameter tractability
(\cref{thm:fpt} is rather a classification result)
but instead rely on mathematical programming to investigate our model of modification fairness.

\paragraph{Setup.}
We computed optimal solutions for \FCE{} using an integer linear programming (ILP) formulation of our problem fed into the commercial solver Gurobi 8.1.1.
The ILP formulation is based on the standard formulation for \textsc{Cluster Editing}~\cite{GW89}, wherein one has a binary variable~$x_{uv}$ for every~$\{u,v\} \in \binom{V(G)}{2}$, indicating whether or~not the solution graph contains the edge~$\{u,v\}$, and three constraints for every vertex triple, which ensure that the triple does not induce a~$P_3$.
The formulation can be easily extended to a formulation for \FCE{} by adding a constraint that ensures that the upper bound~$\delta$ on our fairness~measure~$\diff$ holds.
For every~$\{u,v\} \in \binom{V(G)}{2}$ we set~$\eta_{uv} \coloneqq 1$ if~$\{u, v\} \in E(G)$ and~$\eta_{uv} \coloneqq 0$ otherwise,
and for every~$v \in V(G) = R \uplus B$, we set~$\gamma_v \coloneqq \abs{R}$ if~$v \in R$ and~$\gamma_v \coloneqq -\abs{B}$ if~$v \in B$.
Our ILP formulation now is as follows.
\begin{align*}
	& \text{minimize:} & & \sum_{\{u, v\} \in \binom{V(G)}{2}} \big((1-\eta_{uv}) x_{uv} + \eta_{uv} (1-x_{uv})\big)\\
	& \text{subject to:} &		& + x_{uv} + x_{vw} - x_{uw} \le 1,	\quad\quad\quad\quad\quad\quad \{u, v, w\}	\in \textstyle\binom{V(G)}{3}\\
	&&				& + x_{uv} - x_{vw} + x_{uw} \le 1,	\quad\quad\quad\quad\quad\quad \{u, v, w\}	\in \textstyle\binom{V(G)}{3}\\
	&&				& - x_{uv} + x_{vw} + x_{uw} \le 1,	\quad\quad\quad\quad\quad\quad \{u, v, w\}	\in \textstyle\binom{V(G)}{3}\\
	&&				-\delta \le & \sum_{\{u,v\}\in\binom{V(G)}{2}} \sum_{w \in \{u, v\}} \big((1-\eta_{uv})x_{uv}+\eta_{uv}(1-x_{uv})\big) / \gamma_w \le \delta.
\end{align*}

In order to make the results within the datasets comparable, we introduce a normalized fairness measure~$\Diffnorm(S) \coloneqq \diff(S)/(\nicefrac{2\abs{S}}{\min \{|R|, |B|\}})$.
Clearly, $\Diffnorm(S) \ge 0$, and by our upper bound on~$\diff$ from \cref{obs:bounds}(iii), we have~$\Diffnorm(S) \le 1$, and a solution~$S$ with~$\Diffnorm(S) = 1$ would be maximally unfair.
Analogously, we define~$\diffnorm \coloneqq \delta/(\nicefrac{2k_{\infty}}{\min \{|R|, |B|\}})$, wherein~$k_{\infty}$ is the minimum size of a colorblind ($\delta = \infty$) solution.
Hence, if~$\diffnorm = 0$, then we enforce our solution to be perfectly fair,
whereas, if~$\diffnorm = 1$, then our instance is an instance of standard \textsc{Cluster Editing}, see the discussion after \cref{obs:bounds}.
In \cref{sec:price} we will set~$\delta$ such that it reflects chosen values of~$\diffnorm$.

The experiments were run on machines with an Intel Xeon W-2125 4-core 8-thread CPU clocked at 4.0 GHz and 256GB of RAM, running Ubuntu~18.04.
All material to reproduce the results is publicly available.\footnote{\url{https://git.tu-berlin.de/akt-public/mod-fair-ce}}

For each instance, we set a time limit of one hour for the solving time (excluding the build time).
Whenever Gurobi could not report an optimal solution within that time, we report on the gap obtained in the given time.
Recall that, for a minimization problem (such as ours), the gap is defined as $(z_U-z_L)/z_U$, where~$z_U$ is the smallest (feasible) solution and~$z_L$ is the largest solution lower bound that Gurobi could find within the time limit.
Hence, the gap is at least zero and, whenever a feasible solution was found, at most one.

\paragraph{Datasets.}
We evaluate our model on two datasets.
The first is the SNAP Facebook data set~\cite{snapnets}, which lists for each person (vertex) their gender (color) as well as their friends (edge).
As this data was gathered from Facebook before~2012, the data on gender is binary.
The dataset contains nine graphs.
For each of these graphs and for each~$n \in \{20, 40, 80, 120\}$ and each~$p \in \{0.1, \dots, 0.5\}$, we sampled a subgraph with~$n$ vertices, roughly $pn$ of which were colored red.
As the graphs in the dataset did not have sufficiently large connected components, the goal was to have roughly equally large components.
To this end, among all components but the largest, we chose sufficiently many uniformly at random.
Then, from the largest component, we used a random breadth-first based approach to select a connected subgraph, wherein we randomly selected the next vertex with a bias towards red or towards blue vertices so as to ensure that (roughly) $pn$ vertices were colored red.
For the graphs with~$n = 20, 40$, half of the graphs have one component, and the average number of components is $1.7$ and~$2.6$, respectively.
For the graphs with~$n = 80, 120$, half of the graphs have two components, and the remaining graphs have more components; the average number of components is~$3.0$ and~$3.9$, respectively.
The graphs with~$n = 120$ had a maximum of $9$ components.
Note that, while in standard \CE{}, each connected component can be solved individually, this is not the case for the modification-fair variant as our fairness measure encompasses all components.
We decided not to sample any larger graphs as already our standard solver needed more than half an hour on average to solve standard, colorblind \CE{} on the largest graphs (see \cref{tab:dataset}).
The graphs are grouped into four sets according to~$n$.

The second data set is a product co-review graph based on data from Amazon \cite{amazon}, which was already used to analyze a different fairness model for \CE{} \cite{AEKM20}.
Herein, we have a vertex for each product and an edge whenever two products were reviewed by the same person.
Each vertex belongs to one of five categories out of which we sampled the products.
For each (unordered) pair of product categories,
each~$n \in \{40, 80, 120, 160, 200\}$, and each~$p \in \{0.1, \dots, 0.5\}$,
we sampled subgraphs with roughly~$n$ vertices out of the two categories, roughly $pn$ of which belonged to the one category.
The sampling procedure was slightly different than the one used for the Facebook graphs:
We also used random walks on the smaller connected components so that the proportion of red vertices could be close to~$pn$.
Moreover, we selected the same proportion of vertices from the randomly chosen components.
For each~$n \in \{40, 80, 120, 160, 200\}$, more than half of the sampled graphs had one component.
The average number of components were~$1.5$, $2.0$, $3.0$, $3.9$, and~$4.0$, respectively.
The graphs with~$n \in \{160,200\}$ had a maximum of~$23$ connected components.
Again, the choice in vertices is reflected by the time it took to solve standard, colorblind \CE{} on the graphs (see \cref{tab:dataset}).
We group the graphs into five sets according to~$n$.

\cref{tab:dataset} gives an overview over the graphs, as well as the minimum solution sizes~$k_\infty$ of colorblind solutions and the running time~$t_\infty$ needed to compute said colorblind solutions.
\begin{table}[t]
	\centering
	\caption{Mean of the number of vertices ($n$) and edges ($m$) as well as minimum, mean, and maximum of the solution size~$k_{\infty}$ and running time~$t_{\infty}$ in seconds required to solve standard, colorblind \CE{} on the graphs in our dataset.}
	\label{tab:dataset}
	\small
	\begin{tabular}{l c r r r r r r r r r r r r}
		\toprule
		& & & & \multicolumn{3}{c}{$k_\infty$} & \multicolumn{3}{c}{$t_\infty$ [$s$]}\\
		\cmidrule{5-7}\cmidrule{8-10}
&Set&$n$&$m$&min&mean&max&min&mean&max\\
\midrule
\parbox[t]{5mm}{\multirow{4}{*}{\rotatebox[origin=c]{90}{Facebook}}}
&1& 20.0&  72.2&  8& 28.8&  51&0.0172&0.0327&0.1114\\
&2& 40.0& 230.8& 39&105.4& 185&0.1465&1.3711&12.047\\
&3& 80.0& 902.3&166&432.2& 891&1.2386&267.50&2933.5\\
&4&120.0&1509.8&315&941.4&2992&7.9719&2129.8&3600.1\\
\midrule
\parbox[t]{5mm}{\multirow{5}{*}{\rotatebox[origin=c]{90}{Amazon}}}
&1& 41.7& 109.1& 14& 51.2& 145&0.0489&0.4423&3.3766\\
&2&  83.2& 287.0& 40&151.6& 362&0.5635&16.159&255.12\\
&3& 125.1& 518.1& 61&274.1& 510&5.8775&277.95&3600.1\\
&4&166.8& 808.8& 96&450.2&1294&16.508&1202.5&3600.5\\
&5&207.9&1114.0&107&658.9&1398&58.552&2572.0&3603.0\\
\bottomrule
	\end{tabular}
\end{table}

\subsection{How fair is the colorblind variant?}

We first evaluate the modification fairness of standard, colorblind \CE{}.
\begin{figure}[t]
	\centering
	\includegraphics[width=\columnwidth]{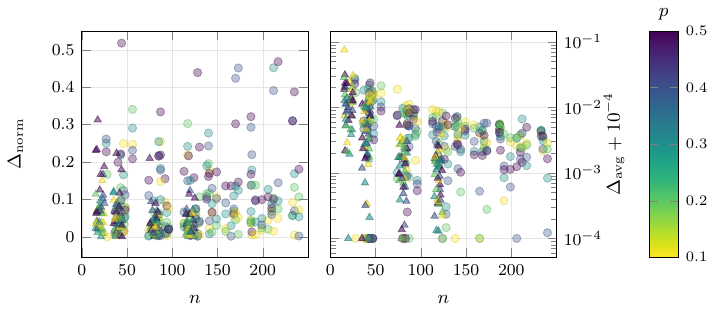}
	\caption{
		\emph{How fair is the ``non-fair'' variant?}
		We compare the normed ($\Diffnorm$) and average ($\Diffavg$) modification fairness of the optimal solution for standard \CE{} to the number~$n$ of vertices and the ratio~$p$ red vertices (color).
		Facebook instances are displayed as triangles, Amazon instances are displayed as circles.
		As the Facebook instances admit four distinct values of~$n$, we add a random jiggle to make similar entries more distinguishable, that is, for Facebook instances, we display~$n+r$ for~$r \in [-5,5]$ chosen uniformly at random.
		The average modification fairness (right) is displayed on a log scale; we add $10^{-4}$ to each entry to make entries with~$\Diffavg = 0$ visible.
	}
	\label{fig:standard-ce-fairness}
\end{figure}
\cref{fig:standard-ce-fairness} shows the number~$n$ of vertices and the modification fairness for each of our instances when run with~$\diffnorm=1$.
Overall, $\Diffnorm$ does not exceed~$0.52$, with~$\Diffnorm < 0.055$ for 50\% and $\Diffnorm < 0.11$ for 75\% of the instances.
For the Amazon graphs, the value is slightly higher with~$\Diffnorm < 0.07$ for 50\% and~$\Diffnorm < 0.15$ for 75\% of the instances,
whereas for Facebook graphs, $\Diffnorm$ does not exceed $0.31$ and is less than~$0.04$ for 50\% and less than~$0.07$ for~75\% of the instances.

In \cref{fig:standard-ce-fairness}, one may observe that
for the Facebook instances with~$n = 120$, $\Diffnorm$ does not exceed 0.13, and for all but one of the Facebook instances with~$n=80$, $\Diffnorm$ is below 0.1.
Indeed,
the means of $\Diffnorm$ are~$0.08$, $0.05$, $0.03$ and $0.03$ for~$n$ being~$20$, $40$, $80$ and $120$, respectively,
that is, $\Diffnorm$ tends to decrease (that is, fairness increases) with increasing number of vertices for the Facebook instances.
The same cannot be observed for Amazon graphs.
Here, the mean of $\Diffnorm$ is more evenly spread,
being~$0.11$, $0.10$, $0.11$, $0.11$ and~$0.12$ for~$n$ being in~$(0,48]$, $(48,96]$, $(96,144]$, $(144,192]$ and $(192,240]$, respectively.
The figure also suggests that there are more outliers among the Amazon instances.
This is backed by the standard deviation, which is~$0.052$ for the Facebook instances and~$0.104$ for the Amazon instances.

In the left plot in \cref{fig:standard-ce-fairness}, one may see that almost all Amazon instances with~$\Diffnorm > 0.3$, especially those with more than~$100$ vertices, have~$p \ge 0.4$.
In the right plot, we evaluate~$\Diffavg(S) \coloneqq \diff(S)/(2\abs{S})$, which norms our fairness measure by the size of the solution
and hence measures for each vertex not the number of incident edits but the percentage of how many of the overall edits are incident to it.
Note that $\Diffavg = \Diffnorm / \min\{|R|, |B|\} \approx \Diffnorm / (np)$, as~$p \le 0.5$ and roughly~$np$ vertices are red.
We cannot derive any correlation between~$p$ and~$\Diffavg$.
We can, however, observe an exponential decay in~$\Diffavg$ with increasing~$n$.
Further, we can observe that the Facebook instances tend to have smaller values of~$\Diffavg$ than Amazon instances with similar numbers of vertices.
Indeed, for instances with~$n \le 75$, the mean values of $\Diffavg$ for Facebook and Amazon instances are~$0.0092$ and~$0.0124$.
For instances with~$75 < n \le 150$, the mean values are~$0.0013$ for Facebook and~$0.0047$ for Amazon.
The Amazon instances with~$n > 150$ have a mean value of~$0.0031$.

For small graphs, the modification fairness is rather low, with $\Delta_{\text{norm}} \ge 0.1$ for 35\% of the graphs in Sets 1 and 2.
For larger graphs, however, even without imposing fairness constraints the solution is already very fair, the mean value of $\Delta_{\text{norm}}$ being 0.05 for graphs in Set 4.
\cref{fig:standard-ce-fairness} further shows that our tested graphs do not allow for a statement whether the initial modification fairness correlates with the ratio between red and blue vertices.

\subsection{The price of fairness}
\label{sec:price}
We next evaluate the \emph{price of fairness}, that is, how much the solution size and the running time increase when requiring the solutions to be fair, wherein we set~$\diffnorm$ to be~$0, 0.01, 0.02, \dots, 0.05$.
As the running time for Set 5 of the Amazon graphs was already very close to our time limit, we evaluate the price of fairness only on Sets 1 to 4 for both Facebook and Amazon graphs.

\begin{table}
	\centering
	\caption{\emph{How much extra time do we need to be fair?} For each of the four sets of the Amazon and Facebook instances,
	we show the mean computation times in seconds for computing \FCE{} with~$\diffnorm$ ranging from~$1$ (colorblind, standard \CE{}) to~$0$ (perfectly fair).
	Recall that the time limit was set to~$3600$ seconds.}
	\label{tab:extra-time}
	\small
{
	\setlength{\tabcolsep}{5pt}
\begin{tabular}{lrrrrrrrrr}
\toprule
$\diffnorm$ & \multicolumn{4}{c}{Amazon Sets 1--4} & \multicolumn{4}{c}{Facebook Sets 1--4}\\
\cmidrule{2-9}
&   \footnotesize $n\approx 40$ &  \footnotesize  $n\approx 80$ &  \footnotesize  $n\approx 120$ & \footnotesize   $n\approx 160$ &\footnotesize $n=20$ & \footnotesize  $n=40$ & \footnotesize  $n=80$ & \footnotesize  $n=120$ \\
\midrule
1                   &    0.44 &   16.16 &  277.95 & 1202.5 & 0.03 &   1.37 &  267.50 & 2129.8 \\
0.05                &    0.49 &   19.88 &  273.15 & 1126.3 & 0.06 &   1.57 &  308.19 & 2174.7 \\
0.04                &    0.54 &   23.75 &  283.39 & 1168.3 & 0.05 &   2.01 &  317.86 & 2168.2 \\
0.03                &    0.56 &   24.28 &  245.19 & 1199.0 & 0.05 &   1.84 &  319.26 & 2196.4 \\
0.02                &    0.80 &   20.27 &  271.08 & 1175.7 & 0.06 &   1.82 &  512.37 & 2179.3 \\
0.01                &    0.81 &   21.16 &  363.61 & 1245.9 & 0.14 &   2.25 &  490.06 & 2281.0 \\
0                &  1223.8 &  3202.0 &  2945.6 & 3188.4 & 1.47 & 601.30 &  2130.3 & 2939.0 \\
\bottomrule
\end{tabular}
}

\end{table}

\cref{tab:extra-time} shows that requiring perfect fairness results in prohibitively high running time.
Allowing a little bit of slack in the fairness however yields significantly lower increments, if any, in running time:
In some cases (such as Sets 3 and 4 of the Amazon graphs), the running time for~$\diffnorm = 0.05$ is even lower than the one for the colorblind ($\diffnorm=1$) instance.
We remark that for all instances of Sets 1 and 2 with~$\diffnorm > 0$ ran within the time limit.
For the instances of Set 3 with~$\diffnorm>0$, the ILP gap did not exceed~$0.03$.
The gaps for the instances of Set 4 as well as the instances with~$\diffnorm=0$ grew as high as~$0.65$.
Of the~$170$ Facebook instances with~$\diffnorm=0$, the solver was able to compute a feasible solution for all but two instances, and ran within the time limit for $65\%$ of the instances.
For the Amazon instances with~$\diffnorm=0$, the solver only computed feasible solutions to~$98$ of the~$178$ instances, for~$26$ of which it was able to run within the time limit.
As for the instances of Set 4 with~$\diffnorm > 0$, $235$ out of~$288$ Amazon instances and~$116$ out of~$240$ Facebook instances were solved within the time limit.

\begin{figure}[t]
	\centering
	\includegraphics[width=\columnwidth]{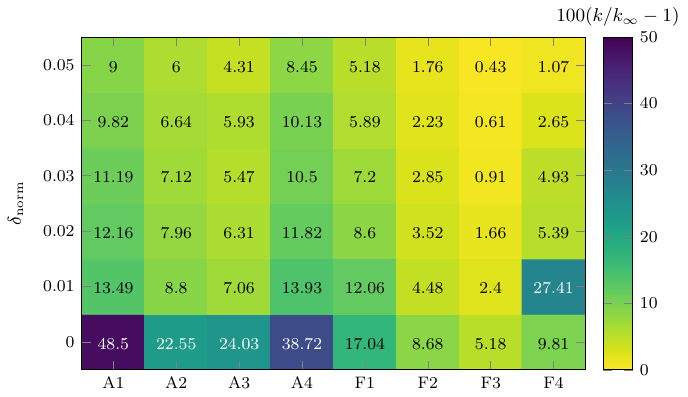}
	\caption{
		\emph{How many extra edits do we need to be fair?}
		Each heat map cell contains the mean percentage by which a minimum solution with fairness of~$\diffnorm$ is larger than the colorblind solution for Sets 1--4 of the Amazon (A) and Facebook (F) graphs, see \cref{tab:dataset}.
	}
	\label{fig:extra-k}
\end{figure}

Let us next consider the percentage by which the solution size~$k$ needs to increase in comparison with the minimum colorblind solution~$k_\infty$.
As one can see in \cref{fig:extra-k}, the solution size needs to increase only slightly, often by less than~$10\%$, in order to obtain a solution that is \emph{almost} fair.
Requiring perfect fairness~$\diffnorm=0$, however, results in a significantly higher increase in solution size.
We can also observe that from Sets 1 to Sets 3, there is a downwards trend in the increase in the solution size, but, for Set 4 instances, the increase in solution size becomes larger again.
To find a possible explanation for this phenomenon, we ask the reader to examine \cref{fig:k-vs-gap}, in which we compare the solution size increase with the colorblind solution size and the ILP gap.
\begin{figure}[t]
	\centering
	\includegraphics[width=\columnwidth]{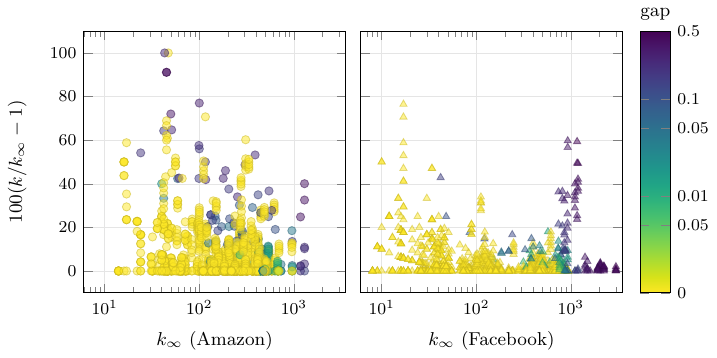}
	\caption{
		We compare for each instance the size of a minimum colorblind solution ($x$-axis) with the percentage by which the minimum fair solution is larger ($y$-axis) for the set of Amazon and Facebook graphs.
		Instances whose percentage was above~$100$ (roughly $1.2\%$ of the Amazon and $2.5\%$ of the Facebook instances) are not displayed.
		The instances are colored by the ILP gap --- note that the coloring follows a logarithmic scale (we added $10^{-3}$ to each gap so as to color zero gaps as well).
	}
	\label{fig:k-vs-gap}
\end{figure}
We can see for both the Amazon and the Facebook graphs that
the fair solution size~$k$ tends to be larger whenever the colorblind solution size~$k_\infty$ is small.
Since~$k_\infty$ correlates with the number~$n$ of vertices (cf.~\cref{tab:dataset}), this fits our findings in \cref{fig:extra-k} for Sets 1 through 3 of both the Amazon and Facebook instances.
A likely reason for this is that, with the smaller solution size, a single edge has a higher impact on the modification fairness, i.e., balancing out the modifications requires proportionally more edits.
A possible answer to why the solution size increases more in graphs with large colorblind solution lies in the gap:
The size of the minimum solution may be smaller than the size of the found (feasible) solution by a fraction of the gap.
This is most evident for the Facebook instances. %
The gap may also be an explanation for the large solution size increase for instances with~$\diffnorm=0$ which we observed in \cref{fig:extra-k}:
As the computation time for the instances with~$\diffnorm=0$ often hit the time limit (cf.~\cref{tab:extra-time}), the gap for these instances may also have been very high in comparison to the gap for the colorblind solution.
Finally, the sudden decrease in the solution size increase for the Facebook instances of Set~4 with~$\diffnorm=0$ compared to the increase with~$\diffnorm=0.01$ is likely due to the fact that for those instances that already were hard to solve with~$\diffnorm=0.01$, our solver could not find a feasible solution for~$\diffnorm=0$ at all.

\paragraph{Discussion.}
There are four main takeaways.
First, on the chosen datasets, the colorblind, standard \CE{} solution seems already rather fair.
Second, the price of fairness in terms of extra running time is low as long as one does not ask for perfect fairness.
However, asking for perfectly fair solutions may become prohibitively expensive.
Third, the price of fairness in terms of solution cost is also low as long as one does not ask for perfect fairness.
Fourth, in all of the above, the ratio between blue and red vertices does not matter very much.

While we can safely state these takeaways from the presented experiments, they are only a first step in evaluating the price of modification fairness.
We would like to conclude this section with some suggestions for future, extended experiments.
First of all, we believe that the experiments are slightly impaired by the fixed time limit.
For future experiments it may be sensible to choose the time limit some constant factor higher than the running time needed to solve colorblind \CE{} on the respective instance.
This would allow for a cleaner analysis of the price of fairness in terms of extra running time.
Also the analysis of the price of fairness in terms of solution size would improve, as the increasing gaps would no longer interfere with the analysis.
Secondly, it would be interesting to study the price of fairness with~$0 < \diffnorm < 0.01$ to figure out the point at which the price of fairness ``explodes''.
Maybe choosing the fairness thresholds from an exponential norm (i.e., $\diffnorm = 10^{-1}, 10^{-2}, 10^{-3}, \dots$) is also more sensible.
More generally, it would be interesting to know which fairness threshold should be considered \emph{reasonable}, or for which fairness threshold one should aim in practice.
These questions come in hand with the more general question of what should be defined as ``fair'', which is a general contentious issue in fairness in algorithms \citet{pessach2023fairness}.

\section{Conclusion}
\label{sec:conclusion}
With our work, we hope to have provided a first step 
towards process-oriented fairness in graph-based data clustering.
Focusing on our newly introduced problem \FCE, there 
are many research challenges.
For instance, in \cref{thm:fpt} we showed that \FCE{} is fixed-parameter 
tractable for the parameter number~$k$ of edge modifications. 
The corresponding exponential factor is~$2^{\bigO(k\log k)}$ --- can we improve on 
this or can we exclude a running time of $2^{o(k\log k)}$ unless the ETH fails?\footnote{We remark that for classic \CE{} there is a tight bound $2^{\Theta(k)}$~\cite{KU12}.}
Further, is \FCE{} parameterized by the number of mono-colored edge modifications~$\mu$ fixed-parameter tractable or \Wone-hard?

A canonical way to continue the studies on \FCE{} is to consider the case of more than two colors.
Indeed, for a constant number of colors, a natural extension of our problem should remain fixed-parameter tractable with respect to the number of edge modifications (cf.~\cref{thm:fpt}):
The number of cliques to keep then depends on the number of colors.
Further, one could study other definitions of fairness over the modifications.

Speaking more generally, one could also combine our process-oriented fairness with other concepts, i.e., the above-mentioned output-oriented fairness \citet{ahmadi2022fair,AEKM20,ahmadian2022improved,schwartz2022correlation}.
Finally, the fairness investigations could be extended to 
generalizations of \CE{} such as \textsc{Hierarchical Tree Clustering}~\cite{GHKNU10}, \textsc{$s$-Plex Cluster Editing}~\cite{GKNU10} or temporal or dynamic versions of \CE{} and related problems, 
e.g.\ \CE{} in temporal graphs~\cite{CMSS18} or dynamic \CE~\cite{LMNN21}.

\bibliographystyle{plainnat}
\bibliography{strings-long,bibliography}

\begin{thebibliography}{31}
\providecommand{\natexlab}[1]{#1}
\providecommand{\url}[1]{\texttt{#1}}
\expandafter\ifx\csname urlstyle\endcsname\relax
  \providecommand{\doi}[1]{doi: #1}\else
  \providecommand{\doi}{doi: \begingroup \urlstyle{rm}\Url}\fi

\bibitem[Abbasi et~al.(2021)Abbasi, Bhaskara, and Venkatasubramanian]{ABV21}
Mohsen Abbasi, Aditya Bhaskara, and Suresh Venkatasubramanian.
\newblock Fair clustering via equitable group representations.
\newblock In \emph{Proceedings of the {ACM} Conference on Fairness,
  Accountability, and Transparency (FAccT~'21)}, pages 504--514. ACM, 2021.
\newblock \doi{10.1145/3442188.3445913}.
\newblock URL \url{https://doi.org/10.1145/3442188.3445913}.

\bibitem[Ahmadi et~al.(2020)Ahmadi, Galhotra, Saha, and
  Schwartz]{ahmadi2022fair}
Saba Ahmadi, Sainyam Galhotra, Barna Saha, and Roy Schwartz.
\newblock Fair correlation clustering, 2020.
\newblock URL \url{https://arxiv.org/abs/2002.03508}.

\bibitem[Ahmadian and Negahbani(2023)]{ahmadian2022improved}
Sara Ahmadian and Maryam Negahbani.
\newblock Improved approximation for fair correlation clustering.
\newblock In Francisco J.~R. Ruiz, Jennifer~G. Dy, and Jan{-}Willem van~de
  Meent, editors, \emph{Proceedings of the 26th International Conference on
  Artificial Intelligence and Statistics ({AISTATS}~'23)}, pages 9499--9516.
  {PMLR}, 2023.
\newblock URL \url{https://proceedings.mlr.press/v206/ahmadian23a.html}.

\bibitem[Ahmadian et~al.(2020{\natexlab{a}})Ahmadian, Epasto, Knittel, Kumar,
  Mahdian, Moseley, Pham, Vassilvitskii, and Wang]{AEKMMPVW20}
Sara Ahmadian, Alessandro Epasto, Marina Knittel, Ravi Kumar, Mohammad Mahdian,
  Benjamin Moseley, Philip Pham, Sergei Vassilvitskii, and Yuyan Wang.
\newblock Fair hierarchical clustering.
\newblock In \emph{Proceedings of the 33rd Annual Coference on Advances in
  Neural Information Processing Systems (NeurIPS~'20)}, pages 21050--21060,
  2020{\natexlab{a}}.
\newblock URL
  \url{https://proceedings.neurips.cc/paper/2020/hash/f10f2da9a238b746d2bac55759915f0d-Abstract.html}.

\bibitem[Ahmadian et~al.(2020{\natexlab{b}})Ahmadian, Epasto, Kumar, and
  Mahdian]{AEKM20}
Sara Ahmadian, Alessandro Epasto, Ravi Kumar, and Mohammad Mahdian.
\newblock Fair correlation clustering.
\newblock In \emph{Proceedings of the 23rd International Conference on
  Artificial Intelligence and Statistics ({AISTATS}~'20)}, pages 4195--4205.
  PMLR, 2020{\natexlab{b}}.
\newblock URL \url{http://proceedings.mlr.press/v108/ahmadian20a.html}.

\bibitem[Bandyapadhyay et~al.(2021)Bandyapadhyay, Fomin, and Simonov]{BFS21}
Sayan Bandyapadhyay, Fedor~V. Fomin, and Kirill Simonov.
\newblock On coresets for fair clustering in metric and euclidean spaces and
  their applications.
\newblock In \emph{Proceedings of the 48th International Colloquium on
  Automata, Languages, and Programming (ICALP~'21)}, pages 23:1--23:15. Schloss
  Dagstuhl - Leibniz-Zentrum f{\"{u}}r Informatik, 2021.
\newblock \doi{10.4230/LIPIcs.ICALP.2021.23}.
\newblock URL \url{https://doi.org/10.4230/LIPIcs.ICALP.2021.23}.

\bibitem[Bandyapadhyay et~al.(2022)Bandyapadhyay, Fomin, Golovach, Purohit, and
  Simonov]{BFGPS21}
Sayan Bandyapadhyay, Fedor~V. Fomin, Petr~A. Golovach, Nidhi Purohit, and
  Kirill Simonov.
\newblock {FPT} approximation for fair minimum-load clustering.
\newblock In \emph{Proceedings of the 17th International Symposium on
  Parameterized and Exact Computation (IPEC~'22)}, pages 4:1--4:14. Schloss
  Dagstuhl - Leibniz-Zentrum f{\"{u}}r Informatik, 2022.
\newblock \doi{10.4230/LIPIcs.IPEC.2022.4}.
\newblock URL \url{https://doi.org/10.4230/LIPIcs.IPEC.2022.4}.

\bibitem[Berger et~al.(2011)Berger, Bonifaci, Grandoni, and
  Sch{\"{a}}fer]{berger2011budgeted}
Andr{\'{e}} Berger, Vincenzo Bonifaci, Fabrizio Grandoni, and Guido
  Sch{\"{a}}fer.
\newblock Budgeted matching and budgeted matroid intersection via the gasoline
  puzzle.
\newblock \emph{Mathematical Programming}, 128\penalty0 (1-2):\penalty0
  355--372, 2011.
\newblock \doi{10.1007/s10107-009-0307-4}.
\newblock URL \url{https://doi.org/10.1007/s10107-009-0307-4}.

\bibitem[B{\"{o}}cker and Baumbach(2013)]{BB13}
Sebastian B{\"{o}}cker and Jan Baumbach.
\newblock Cluster editing.
\newblock In \emph{Proceedings of the 9th International Conference on
  Computability in Europe (CiE~'13)}, pages 33--44. Springer, 2013.
\newblock \doi{10.1007/978-3-642-39053-1\_5}.
\newblock URL \url{https://doi.org/10.1007/978-3-642-39053-1\_5}.

\bibitem[B{\"{o}}cker et~al.(2011)B{\"{o}}cker, Briesemeister, and Klau]{BBK11}
Sebastian B{\"{o}}cker, Sebastian Briesemeister, and Gunnar~W. Klau.
\newblock Exact algorithms for cluster editing: {E}valuation and experiments.
\newblock \emph{Algorithmica}, 60\penalty0 (2):\penalty0 316--334, 2011.
\newblock \doi{10.1007/s00453-009-9339-7}.
\newblock URL \url{https://doi.org/10.1007/s00453-009-9339-7}.

\bibitem[Cai(1996)]{Cai96}
Leizhen Cai.
\newblock Fixed-parameter tractability of graph modification problems for
  hereditary properties.
\newblock \emph{Information Processing Letters}, 58\penalty0 (4):\penalty0
  171--176, 1996.
\newblock \doi{10.1016/0020-0190(96)00050-6}.
\newblock URL \url{https://doi.org/10.1016/0020-0190(96)00050-6}.

\bibitem[Chakrabarty and Negahbani(2021)]{CN21}
Deeparnab Chakrabarty and Maryam Negahbani.
\newblock Better algorithms for individually fair $k$-clustering.
\newblock In \emph{Proceedings of the 34th Annual Coference on Advances in
  Neural Information Processing Systems (NeurIPS~'21)}, pages 13340--13351,
  2021.
\newblock URL
  \url{https://proceedings.neurips.cc/paper/2021/hash/6f221fcb5c504fe96789df252123770b-Abstract.html}.

\bibitem[Chen et~al.(2006)Chen, Huang, Kanj, and Xia]{CHKX06}
Jianer Chen, Xiuzhen Huang, Iyad~A. Kanj, and Ge~Xia.
\newblock Strong computational lower bounds via parameterized complexity.
\newblock \emph{Journal of Computer and System Sciences}, 72\penalty0
  (8):\penalty0 1346--1367, 2006.
\newblock \doi{10.1016/j.jcss.2006.04.007}.
\newblock URL \url{https://10.1016/j.jcss.2006.04.007}.

\bibitem[Chen et~al.(2018)Chen, Molter, Sorge, and Such{\'{y}}]{CMSS18}
Jiehua Chen, Hendrik Molter, Manuel Sorge, and Ondrej Such{\'{y}}.
\newblock Cluster editing in multi-layer and temporal graphs.
\newblock In \emph{Proceedings of the 29th International Symposium on
  Algorithms and Computation (ISAAC~'18)}, pages 24:1--24:13. Schloss Dagstuhl
  - Leibniz-Zentrum f{\"{u}}r Informatik, 2018.
\newblock \doi{10.4230/LIPIcs.ISAAC.2018.24}.
\newblock URL \url{https://doi.org/10.4230/LIPIcs.ISAAC.2018.24}.

\bibitem[Chierichetti et~al.(2017)Chierichetti, Kumar, Lattanzi, and
  Vassilvitskii]{chierichetti2017clustering}
Flavio Chierichetti, Ravi Kumar, Silvio Lattanzi, and Sergei Vassilvitskii.
\newblock Fair clustering through fairlets.
\newblock In \emph{Proceedings of the 30th Annual Coference on Advances in
  Neural Information Processing Systems (NIPS~'17)}, pages 5029--5037. Curran
  Associates, Inc., 2017.
\newblock URL \url{https://papers.nips.cc/paper/by-source-2017-2591}.

\bibitem[Friggstad and Mousavi(2021)]{FM21}
Zachary Friggstad and Ramin Mousavi.
\newblock Fair correlation clustering with global and local guarantees.
\newblock In \emph{Proceedings of the 17th International Symposium on
  Algorithms and Data Structures (WADS~'21)}, pages 414--427. Springer, 2021.
\newblock \doi{10.1007/978-3-030-83508-8\_30}.
\newblock URL \url{https://doi.org/10.1007/978-3-030-83508-8\_30}.

\bibitem[Garey and Johnson(1975)]{GJ75}
Michael~R. Garey and David~S. Johnson.
\newblock Complexity results for multiprocessor scheduling under resource
  constraints.
\newblock \emph{SIAM Journal on Computing}, 4:\penalty0 397--411, 1975.
\newblock \doi{10.1137/0204035}.
\newblock URL \url{https://doi.org/10.1137/0204035}.

\bibitem[Ghadiri et~al.(2021)Ghadiri, Samadi, and Vempala]{GSV21}
Mehrdad Ghadiri, Samira Samadi, and Santosh~S. Vempala.
\newblock Socially fair $k$-means clustering.
\newblock In \emph{Proceedings of the {ACM} Conference on Fairness,
  Accountability, and Transparency (FAccT~'21)}, pages 438--448. ACM, 2021.
\newblock \doi{10.1145/3442188.3445906}.
\newblock URL \url{https://doi.org/10.1145/3442188.3445906}.

\bibitem[Gr{\"{o}}tschel and Wakabayashi(1989)]{GW89}
Martin Gr{\"{o}}tschel and Yoshiko Wakabayashi.
\newblock A cutting plane algorithm for a clustering problem.
\newblock \emph{Mathematical Programming}, 45\penalty0 (1-3):\penalty0 59--96,
  1989.
\newblock \doi{10.1007/BF01589097}.
\newblock URL \url{https://doi.org/10.1007/BF01589097}.

\bibitem[Guo et~al.(2010{\natexlab{a}})Guo, Hartung, Komusiewicz, Niedermeier,
  and Uhlmann]{GHKNU10}
Jiong Guo, Sepp Hartung, Christian Komusiewicz, Rolf Niedermeier, and Johannes
  Uhlmann.
\newblock Exact algorithms and experiments for hierarchical tree clustering.
\newblock In \emph{Proceedings of the 24th Conference on Artificial
  Intelligence (AAAI~'10)}, pages 457--462. AAAI Press, 2010{\natexlab{a}}.
\newblock \doi{10.1609/aaai.v24i1.7684}.
\newblock URL \url{https://doi.org/10.1609/aaai.v24i1.7684}.

\bibitem[Guo et~al.(2010{\natexlab{b}})Guo, Komusiewicz, Niedermeier, and
  Uhlmann]{GKNU10}
Jiong Guo, Christian Komusiewicz, Rolf Niedermeier, and Johannes Uhlmann.
\newblock A more relaxed model for graph-based data clustering: $s$-plex
  cluster editing.
\newblock \emph{SIAM Journal on Discrete Mathematics}, 24\penalty0
  (4):\penalty0 1662--1683, 2010{\natexlab{b}}.
\newblock \doi{10.1137/090767285}.
\newblock URL \url{https://doi.org/10.1137/090767285}.

\bibitem[Komusiewicz and Uhlmann(2012)]{KU12}
Christian Komusiewicz and Johannes Uhlmann.
\newblock Cluster editing with locally bounded modifications.
\newblock \emph{Discrete Applied Mathematics}, 160\penalty0 (15):\penalty0
  2259--2270, 2012.
\newblock \doi{10.1016/j.dam.2012.05.019}.
\newblock URL \url{https://doi.org/10.1016/j.dam.2012.05.019}.

\bibitem[Leskovec and Krevl(2014)]{snapnets}
Jure Leskovec and Andrej Krevl.
\newblock {SNAP Datasets}: {S}tanford large network dataset collection, 2014.
\newblock URL \url{http://snap.stanford.edu/data}.

\bibitem[Leskovec et~al.(2007)Leskovec, Adamic, and Huberman]{amazon}
Jure Leskovec, Lada~A. Adamic, and Bernardo~A. Huberman.
\newblock The dynamics of viral marketing.
\newblock \emph{ACM Transactions on the Web}, 1\penalty0 (1):\penalty0 5, 2007.
\newblock \doi{10.1145/1232722.1232727}.
\newblock URL \url{https://doi.org/10.1145/1232722.1232727}.

\bibitem[Luo et~al.(2021)Luo, Molter, Nichterlein, and Niedermeier]{LMNN21}
Junjie Luo, Hendrik Molter, Andr{\'{e}} Nichterlein, and Rolf Niedermeier.
\newblock Parameterized dynamic cluster editing.
\newblock \emph{Algorithmica}, 83\penalty0 (1):\penalty0 1--44, 2021.
\newblock \doi{10.1007/s00453-020-00746-y}.
\newblock URL \url{https://doi.org/10.1007/s00453-020-00746-y}.

\bibitem[Mahabadi and Vakilian(2020)]{MV20}
Sepideh Mahabadi and Ali Vakilian.
\newblock Individual fairness for $k$-clustering.
\newblock In \emph{Proceedings of the 37th International Conference on Machine
  Learning (ICML~'20)}, volume 119, pages 6586--6596. PMLR, 2020.
\newblock URL \url{http://proceedings.mlr.press/v119/mahabadi20a.html}.

\bibitem[Mehrabi et~al.(2022)Mehrabi, Morstatter, Saxena, Lerman, and
  Galstyan]{mehrabi2022bias}
Ninareh Mehrabi, Fred Morstatter, Nripsuta Saxena, Kristina Lerman, and Aram
  Galstyan.
\newblock A survey on bias and fairness in machine learning.
\newblock \emph{ACM Computing Surveys}, 54\penalty0 (6):\penalty0
  115:1--115:35, 2022.
\newblock \doi{10.1145/3457607}.
\newblock URL \url{https://doi.org/10.1145/3457607}.

\bibitem[Mulmuley et~al.(1987)Mulmuley, Vazirani, and Vazirani]{MulmuleyVV87}
Ketan Mulmuley, Umesh~V. Vazirani, and Vijay~V. Vazirani.
\newblock Matching is as easy as matrix inversion.
\newblock \emph{Combinatorica}, 7\penalty0 (1):\penalty0 105--113, 1987.
\newblock \doi{10.1007/BF02579206}.
\newblock URL \url{https://doi.org/10.1007/BF02579206}.

\bibitem[Pessach and Shmueli(2023)]{pessach2023fairness}
Dana Pessach and Erez Shmueli.
\newblock A review on fairness in machine learning.
\newblock \emph{ACM Computing Surveys}, 55\penalty0 (3):\penalty0 51:1--51:44,
  2023.
\newblock \doi{10.1145/3494672}.
\newblock URL \url{https://doi.org/10.1145/3494672}.

\bibitem[Schwartz and Zats(2022)]{schwartz2022correlation}
Roy Schwartz and Roded Zats.
\newblock Fair correlation clustering in general graphs.
\newblock In \emph{Proceedings of the Conference on Approximation,
  Randomization, and Combinatorial Optimization ({APPROX/RANDOM}~'22)}, pages
  37:1--37:19. Schloss Dagstuhl - Leibniz-Zentrum f{\"{u}}r Informatik, 2022.
\newblock \doi{10.4230/LIPIcs.APPROX/RANDOM.2022.37}.
\newblock URL \url{https://doi.org/10.4230/LIPIcs.APPROX/RANDOM.2022.37}.

\bibitem[Vakilian and Yal\c{c}{\i}ner(2021)]{VY21}
Ali Vakilian and Mustafa Yal\c{c}{\i}ner.
\newblock Improved approximation algorithms for individually fair clustering,
  2021.
\newblock URL \url{https://arxiv.org/abs/2106.14043}.

\end{thebibliography}

\ifappendix{}
\clearpage
\appendix
\section*{Appendix}
\appendixProofText
\fi{}

\end{document}